\documentclass[a4paper,onecolumn,11pt,accepted=2025-10-06]{quantumarticle}
\pdfoutput=1

\usepackage{amsmath, amsthm, amssymb, amsfonts}
\usepackage[utf8]{inputenc} 
\usepackage[T1]{fontenc}    
\usepackage{hyperref}       
\usepackage{url}            
\usepackage{nicefrac}       
\usepackage{microtype}      
\usepackage{lipsum}
\usepackage{graphicx}       
\usepackage{multirow}
\usepackage[table]{xcolor}
\usepackage{svg} 

\usepackage[numbers]{natbib}
\usepackage{comment}
\usepackage{physics}
\theoremstyle{definition}
\usepackage{mathtools}

\newtheorem{theorem}{Theorem}
\newtheorem{lemma}[theorem]{Lemma}
\newtheorem{corollary}[theorem]{Corollary}
\newtheorem{definition}[theorem]{Definition}
\newtheorem{example}[theorem]{Example}
\newtheorem{proposition}[theorem]{Proposition}
\theoremstyle{remark}
\newtheorem{remark}[theorem]{Remark}

\newtheorem{algorithm}{Algorithm}
\usepackage{graphicx}
\usepackage{changepage}
\newcommand{\umd}{QuICS, University of Maryland, College Park, MD 20742, USA.}
\begin{document}

\title{Error Correction in Dynamical Codes}
\author{Esther Xiaozhen Fu}
\affiliation{\umd}
\email{xz1@umd.edu}
\author{Daniel Gottesman}
\email{dgottesman@umd.edu}
\affiliation{\umd}
\affiliation{Computer Science Department, University of Maryland, College Park, MD 20742, USA}

\begin{abstract}
    We ask what is the general framework for a quantum error correcting code that is defined by a sequence of measurements. Recently, there has been much interest in Floquet codes and space-time codes. In this work, we define and study the distance of a dynamical code. This is a subtle concept and difficult to determine: At any given time, the system will be in a subspace which forms a quantum error-correcting code with a given distance, but the full error correction capability of that code may not be available due to the schedule of measurements associated with the code. We address this challenge by developing an algorithm that tracks information we have learned about the error syndromes through the protocol and put that together to determine the distance of a dynamical code, in a non-fault-tolerant context. We use the tools developed for the algorithm to analyze the initialization and masking properties of a generic Floquet code. Further, we look at properties of dynamical codes under the constraint of geometric locality with a view to understand whether the fundamental limitations on logical gates and code parameters imposed by geometric locality for traditional codes can be surpassed in the dynamical paradigm. We find that codes with a limited number of long range connectivity will not allow non-Clifford gates to be implemented with finite depth circuits in the 2D setting. 

\end{abstract}

\maketitle

\section{Introduction}

Quantum error correcting codes form the foundation of scalable quantum computing. The performance and experimental feasibility of codes depend on their properties. Particularly important among them are code distance, logical gates and tradeoffs between the number of logical qubits and distance. In this paper, we study these code properties in the context of the new Floquet code paradigm. 

Floquet codes extend the framework of conventional “static” stabilizer codes--those defined by fixed stabilizer or gauge groups--by introducing time-dependent measurements. The honeycomb code introduced in \cite{Hastings2021dynamically,haah_boundaries_2022} is a first example of a new paradigm of quantum codes that features periodic sequence of measurements that anti-commute with the previous round of measurements, leading to a periodic change in both the stabilizer group and the logical representatives. At each moment, the stabilizers form what's known as the \textit{instantaneous stabilizer group} (ISG). Decoding involves combining measurement outcomes from multiple rounds. One key advantage of Floquet code is the low weight measurements, which are combined to reveal the syndrome of high-weight stabilizers—unlocking the power of error correction with minimal overhead. 

Since the first paper on Floquet codes, numerous additional examples of Floquet codes have been found. This includes planar Floquet codes \cite{vuillot2021planar,haah_boundaries_2022,kesselring2022anyon,ellison2023floquet}, 2D Floquet codes defined on hyperbolic surfaces\cite{fahimniya2023faulttolerant,higgott2023constructions}, and 3D Floquet codes\cite{PhysRevB.108.205116,bauer2023topological,davydova_floquet_2023}. There are also CSS version of Floquet codes and codes with non-periodicity \cite{davydova_floquet_2023,kesselring2022anyon}. The threshold and performance of Floquet codes have been demonstrated in \cite{Gidney_2021,PRXQuantum.4.010310}. Progress has also been made in the theoretical understanding of Floquet codes, including results based on measurement quantum cellular automata \cite{aasen2023measurement}, adiabatic paths of gapped Hamiltonians \cite{aasen_adiabatic_2022}, and unitary loops \cite{sullivan2023floquet}. Later, new Floquet codes that are not based on colorable graphs have also been found, including the Floquet Bacon-Shor code \cite{alam2024} and Floquet Haah code \cite{xu2025}. 

Generalizing and broadening the idea of Floquet codes, we wish to understand the framework of quantum codes that are defined by a sequence of measurements. Our study is in part driven by insights on code deformation and code switching \cite{huang2018transversal,kubica_universal_2015} in implementing a universal set of gates: For example, code switching between Reed Muller code and Steane code enables the transversal implementation of universal logical gates \cite{anderson_fault-tolerant_2014}. As a dynamical code consists of several rounds of ISGs, it can potentially circumvent the no-go theorem by Eastin and Knill \cite{eastin_restrictions_2009} which shows that transversal operators on any non-trivial quantum code belong to a finite group, and thus cannot be universal.  We use the term dynamical code instead of Floquet code, as we do not exclude the possibility of a non-periodic sequence of Pauli measurements. Most previous works of Floquet codes have concentrated on codes defined on three-colorable graphs, here we look at more general codes: A dynamical code is more generally defined as a code that involves rounds of Pauli measurements. Each round of measurements consists of only commuting Pauli operators. In particular, a two round dynamical code is a subsystem code, such as the Bacon-Shor code \cite{bravyi_subsystem_2011}, which consists of geometrically local gauges and non-local stabilizers.  

In this work, we present a comprehensive study of the framework of dynamical code using QEC language. We formalize the concept of masking with respect to a dynamical code and introduce a notion of distance that generalizes known notions of distance of stabilizer and subsystem codes. Building on this theoretical foundation, we will also introduce an efficient algorithm to determine the syndromes and therefore the distance for a given dynamical code. Making use of the structure of this algorithm, we derive new insights into the initialization of arbitrary Floquet codes. Finally, we establish fundamental limitations by extending the Bravyi--K\"{o}nig no-go theorem to the setting of dynamical codes. See Section \ref{sec:overview} for an overview of the main results. 

The paper is structured as follows. We start with an overview of our results in Section \ref{sec:overview}. In Section \ref{sec:preliminary}, we give a brief introduction to quantum error correcting codes and define some notations for the rest of the paper. In Section \ref{sec: dynamical code}, we define distance for dynamical codes and discuss the logical operators and syndrome information in the presence of measurements. Section \ref{sec: the distance algorithm} discusses the distance algorithm: In Subsection \ref{subsec:intuition}, we give an intuition through examples to demonstrate that it can be tricky to find the distance of a given round in the dynamical code. In Subsection \ref{subsec: algorithm}, we present the algorithm for finding the unmasked distance of dynamical codes, and show that the algorithm can output all the unmasked stabilizers and permanently masked stabilizers of an ISG. In Section \ref{sec:init}, we apply the tools developed for the distance algorithm to analyze the initialization and masking properties of Floquet codes. In Section \ref{sec: geometrically-local-codes}, we generalize previous works for geometrically local stabilizer codes to dynamical codes and also present a no-go theorem for non-Clifford transversal gates on 2D geometrically local dynamical codes. 

\section{Overview}\label{sec:overview}
We now give an overview of our contributions in this work.
\paragraph{Framework of dynamical code}
In Section \ref{sec: dynamical code}, we extend the framework of quantum error correcting codes from stabilizer and subsystem codes to dynamical codes. 

Subsection \ref{subsec: distance} present a tighter notion of distance for dynamical code. Our definition captures the fact that, unlike static codes, the correctability of errors in a dynamical code depends not only on the current stabilizer group but also on the syndrome information recoverable from future measurements. The distance of a dynamical code is defined as the minimum unmasked distance over all instantaneous stabilizer groups (ISGs), where the unmasked distance quantifies the set of correctable errors under the given measurement sequence.

What makes the previous notions of distance inadequate for dynamical codes is that a measurement scheme can irreversibly remove some syndrome information from the code while other parts of the syndromes are revealed. This is the concept of masking which has been previously introduced by Gottesman \cite{gottesman2022opportunities}. For instance, a stabilizer generator can be permanently masked if its eigenvalue cannot be obtained again in the future. These are elements that have been irreversibly removed without revealing their syndrome information. 

We classify each stabilizer generator in every instantaneous stabilizer group into three categories: 

\begin{itemize}
    \item Permanently masked stabilizers: Stabilizers whose syndromes cannot be obtained using given sequence of measurements.
    \item Temporarily masked stabilizers: Stabilizers whose syndromes have yet to be obtained with the current sequence but can still be obtained if we choose further measurements carefully.
    \item Unmasked stabilizers: Stabilizers whose syndromes can be obtained from outcomes of the current sequence of measurements. 
\end{itemize}

Using this classification, we define the unmasked distance of an ISG in Definition \ref{def:unmasked distance}. We prove that this notion yields a tighter upper bound on the set of correctable errors than prior definitions, by generalizing the Knill–Laflamme condition to the setting where only partial syndrome information is available. We conjecture that the unmasked distance is tight in the absence of measurement errors. 

In Subsection \ref{subsec: Logical operators}, we provide a general theoretical framework for analyzing how logical operators and logical errors evolve over time in a dynamical code. Unlike static codes, where logical representatives are fixed, in a dynamical code the logical representatives and their values are inherently spacetime-dependent: the outcome of a logical operator can depend both on the spacetime location of errors and on the measurement outcomes of relevant stabilizers at different times.

Theorem \ref{theorem: spacetime errors} in Subsection \ref{subsec: syndromes through} shows that, in addition to where the errors happen, their temporal positions crucially affect the observed syndromes. This result highlights the need for decoding algorithms that take into account not just which qubits are erroneous, but also when those errors occur, relative to the measurement sequence.

\paragraph{Algorithm for finding distance of a dynamical code}

In Section \ref{sec: the distance algorithm}, we present the distance algorithm. 

The distance algorithm is an efficient classical algorithm that determines the set of all the syndrome information that can be learnt from a fixed measurement window (see Section~\ref{sec: the distance algorithm}). The objective is to classify the stabilizer generators of a fixed ISG into three types: \emph{unmasked}, \emph{temporarily masked}, and \emph{permanently masked}. This classification captures which syndromes can be reliably extracted in a finite time window and provides an upper bound—the \emph{unmasked distance}—on the distance of the code (Definition~\ref{def:unmasked distance}).

The algorithm proceeds by tracking the evolution of stabilizers and measurement outcomes through a set of commutation-based update rules (Subsection~\ref{subsec: algorithm}). Two sets, \( C \) and \( V \), are iteratively updated to record the transformed stabilizer generators and accumulated measurements, respectively. The unmasked stabilizers are obtained by computing the intersection \( \langle C \rangle \cap \langle V \rangle \), implemented via the Zassenhaus algorithm (Algorithm~\ref{alg: Zassenhaus algorithm}). From this, the syndromes for unmasked stabilizers are reconstructed (Algorithm~\ref{alg: unmasked stabilizers}). The algorithm also determines permanently masked stabilizers and their destabilizers (Algorithm~\ref{alg: pms and destabilizer}), as well as the temporarily masked ones (Algorithm~\ref{alg: tms}).

Our framework provides a concrete and efficient tool to determine the partial decoding power of a dynamical code under realistic measurement constraints. The algorithm runs in \( O(n^3) \) time, with the main cost coming from the subroutine to compute generator intersections. Several illustrative examples are given in Subsection~\ref{subsec:examples}.
 
The distance algorithm has immediate application to code construction for dynamical codes. For example, it can be used to check the unmasked distance of existing examples of dynamical codes and design measurement sequences to optimize the distance. The framework in the distance algorithm can be of independent interest, as it is a useful tool to analyze other dynamical code properties. In particular, we utilize it to show several results on Floquet codes with periodic measurement sequence. 

Going beyond code construction, the distance algorithm can be extended with relative ease to obtain the outcome code in \cite{delfosse2023spacetime} since the unmasked stabilizers are essentially detector cells restricted to the time window being analyzed, and that Clifford operators can be included between measurement rounds. But beyond this, we can use the distance algorithm to obtain gauge elements of a subsystem spacetime code, since we also take into account other elements including the permanently and temporarily masked stabilizers with their destabilizers. We also track which round(s) of the ISGs the syndrome information is for, and the time window one requires to obtain enough syndrome information for a given round. Although we consider errors in round 0 and assume perfect measurements (in the presence of measurement errors, it is possible for a smaller set of errors in space-time to result in a logical error \cite{vuillot2021planar}), our results can contribute to a better understanding of dynamical codes by establishing some code properties that are inherent to dynamical codes. The results here can also be extended to analyze the fault-tolerance aspect of dynamical codes and then to that of spacetime codes. A better understanding of the distance for dynamical codes will be insightful for understanding the distance and fault tolerance for a Clifford circuit, even eventually extending to include more general noise models. 

One aspect of our algorithm is that it identifies which stabilizer generators are temporarily masked and which are permanently masked.  This is critical information for a dynamical coding protocol that is described in a modular way: Temporarily masked generators from one module may usefully contribute syndrome information into a later module whereas syndrome information from permanently masked generators will never be useful, regardless of the choice of later modules.

Lastly, we remark that there have been several related works. The distance algorithm helps complete the picture by adding anti-commuting measurements into the framework of \cite{bacon_sparse_2017, gottesman2022opportunities}, which translate circuits into quantum codes. The syndrome information from the distance algorithm contains the same information as the algorithm in \cite{delfosse2023spacetime}. The connection between understanding the detector cell/spacetime code picture to designing fault tolerant circuits has also been explored in \cite{derks2024}, and the concept of fault distance (which is related to spacetime code distance) is explored in \cite{beverland2024faulttolerance}.

\paragraph{Initialization of a Floquet Code}  
In Section~\ref{sec:init}, we provide a detailed study of the initialization behavior of Floquet codes. Each cycle of a Floquet code involves measuring the same set of operators in a fixed order. Understanding this initialization process turns out to be fruitful in understanding how many rounds it takes to do a full round of syndrome measurements. 

Our first key result (Theorem~\ref{theorem: subset_floquet}) shows that the set of instantaneous stabilizer generators grows monotonically over successive cycles: specifically, after measuring a given operator $m_i$, the set of stabilizers present at that point in cycle $j$ is always a subset of those in cycle $j+1$. This result is established by applying the machinery of the \textit{distance algorithm} introduced earlier in the paper. The use of this algorithm provides clarity in proving the results that would otherwise require much more elaborate tracking of stabilizers.

Building on this, Theorem~\ref{theorem: new_floquet} formalizes how the number of new stabilizer generators can evolve from cycle to cycle. We show that if $k$ new generators are added in a given cycle, then at most $k$ additional generators may appear in the subsequent cycle. Moreover, for any specific measurement $m_i$, a growth in the number of stabilizer generators in the current cycle implies that the same growth must have been possible in the previous cycle. 

To explore the limits of how many cycles it take to fully initialize, we construct an explicit example showing that this bound is tight: that is, there exist Floquet codes that require exactly $n{-}1$ cycles to initialize a code with $n$ stabilizer generators. This is accomplished via a recursive construction of measurement sequences that add exactly one new stabilizer per cycle. Figures~\ref{fig: init_floquet} and~\ref{fig:seq_insert_floquet} illustrate the mechanics of this recursive design, showing how existing stabilizers are transformed and extended to generate the fully initialized set after a very long number of iterations of a measurement sequence.

Crucially, this initialization behavior is tightly linked to the understanding \textit{masking} and syndrome extraction. The final theorem in this section connects the number of cycles needed for full initialization with the time required to unmask all the stabilizer generators in a given instantaneous stabilizer group (ISG). We prove that after at most $k$ cycles—where $k$ is the number of cycles required to initialize the full Floquet code—all stabilizers in the ISG will be either fully unmasked or permanently masked. This connection allows us to bound the total number of required measurements to determine all unmasked stabilizers: if each cycle involves $m$ measurements and the code requires $k$ cycles to initialize, then the number of measurements needed to fully determine the unmasked part of any ISG is at most $m \cdot k$. In particular, for code families with constant initialization depth and measurements ($k = O(1)$, $m = O(1)$), the unmasking process remains efficient and scalable.

\paragraph{Generalized Bravyi-Koenig theorem to dynamical code} (Section \ref{sec: geometrically-local-codes})
For a quantum code with a fixed set of stabilizer generators, the problem of finding the distance and fault-tolerant gates has been previously studied from the angle of embedding them in Euclidean space, as \textit{locality} is very important in practical code design. Bravyi and Terhal \cite{bravyi_no-go_2009} proved an upper bound on local code distance: $d = O(L^{D-1})$. Further, Bravyi and K{\"o}nig \cite{bravyi_classification_2013} showed that a $D$-dimensional local code can only transversally implement encoded gates from the set $\mathcal{C}^{(D)}$, with $\mathcal{C}^{(1)}$ the Pauli group and $\mathcal{C}^{(j)}$ as the $j^{\mathrm{th}}$ level of the Clifford hierarchy. Thus, we are interested in answering this \textbf{question}: To what extent does these locality restrictions apply to geometrically local dynamical codes?

We answer the above question and also generalize Pastawski and Yoshida's results \cite{pastawski_fault-tolerant_2015} to give a no-go theorem for 2D geometrically local dynamical code to support transversal non-Clifford gates even with some amount of long-range connectivity and non-local stabilizers. This result also extends to higher dimensions. Specifically, we prove that if a logical qubit $q_L$ can be supported on a region sufficiently far away from the qubits with long range connectivity, then a transversal gate that implements a logical single qubit gate on $q_L$ must be from the $D^{\mathrm{th}}$ level Clifford hierarchy. Furthermore, if there is limited amount of long range connectivity, then the code cannot support a logical gate from a higher Clifford hierarchy.

\section{Background and Notation}
\label{sec:preliminary}

\subsection{Stabilizer code}
\begin{definition}
A stabilizer code $\mathcal{Q}=\mathcal{Q}(S)$ with parameters $[[n,k,d]]$ is a code defined on the Hilbert space $\mathcal{H}$ of $n$ qubits. It has $n-k$ Pauli stabilizer generators $S$ and $k$ logical qubits. The codeword space $\mathcal{C}$ is a subspace of $\mathcal{H}$, and contains the $+1$ eigenvectors for stabilizers in $S$. The non-trivial logical representatives are given by the set $N(S)\backslash \langle S\rangle$, where $N(S)$ is the normalizer of $\langle S\rangle$ defined as $ N(S) = \{p \in \mathcal{P}_n: \langle S\rangle p = p\langle S\rangle \}, \mathcal{P}_n$ is the $n$ qubit Pauli group, with $\mathcal{P}=\{\pm 1, \pm i, I,X,Y,Z\}$, $\mathcal{P}_n = \mathcal{P}^{\otimes n}$ and $\langle S \rangle$ is the group generated by $S$. $d$ is the distance, which for a stabilizer code is the smallest weight Pauli operator $L$ for which $L$ commutes with all elements of $S$ but is not itself in $\langle S\rangle$. The distance is the lowest weight of elements in $N(S) \backslash \langle S\rangle$.

\end{definition}

A code is a low-density parity-check (LDPC) code if each generator acts on a constant number of qubits and each qubit is involved in a constant number of generators. 

\subsection{Clifford Hierarchy}
 
 The Clifford hierarchy given by $\{\mathcal{C}^{(k)}, k\geq 1\}$ is defined recursively. $\mathcal{P}_n  $ forms the first level of the Clifford hierarchy and is denoted by  $ \mathcal{C}^{(1)} $. The second level of the Clifford hierarchy $ \mathcal{C}^{(2)} $ is the Clifford group and is given by the group of automorphisms of the Pauli group: $ \{U\mid UPU^{\dagger}\in \mathcal{C}^{(1)}, \forall P \in \mathcal{P}_n\} $. Higher levels in the hierarchy are defined recursively in the same way: 
 \begin{equation}
 	\mathcal{C}^{(k)} := \{U\mid UPU^{\dagger}\in \mathcal{C}^{(k-1)}, \forall P \in \mathcal{P}_n\}
 \end{equation}

\subsection{Outcome and Syndrome Information}\label{subsec: outcome and syndrome}
The outcome of a measurement $m$ is denoted by $O(m) \in \{\pm 1\}$. The outcome $O(s) \in \{\pm 1\}$ is the value of $s \in \langle S\rangle$ and in general can be obtained through the product of outcomes of measurements from different rounds of a measurement sequence. 

The syndrome for an error $E$ (which can consist of errors that occurred at different times in a sequence of measurements) for a particular stabilizer $s$ in an ISG is given by the symplectic inner product $E \odot s$ which is defined below. For a set of stabilizers with initial outcomes and a set of final outcomes at a later time, the difference in the two sets of outcomes gives the syndrome information which one can use to decode and output the best guess of the error in this time window whose syndrome matches the syndrome information that one has obtained through the measurement sequence. 

To define the symplectic inner product between two Pauli operators $A$ and $B$ in $\mathcal{P}_n$, which we denote as $A \odot B$, we first map the Pauli operators to their binary representation in a $2n$ dimensional vector space $V$: 
\begin{equation}\label{eq: pauli to binary}
    F: \mathcal{P}_n \to V, p \mapsto [a,b],
\end{equation}

Here $a$ and $b$ are row vectors with n entries, with $a = [a_1, a_2,\cdots,a_n]$, $b=[b_1,b_2,\cdots,b_n]$, such that $p$ can be written as $p =i^{ab^{T}} \prod_jX_j^{a_j} Z_j^{b_j}$ up to some phases. 
\begin{definition}
     If the vectors for $A$ and $B$ are mapped to $[a,b]$ and $[c,d]$ respectively, the symplectic inner product $A \odot B$ is given by 
    \begin{equation}
      A \odot B=   [a,b]\begin{pmatrix}
0 & I_n \\
I_n & 0 
\end{pmatrix}
[c,d]^{T} \mathrm{(mod\: 2)}
    \end{equation}
    \end{definition}

\subsection{Update rules}
For a dynamical code, its stabilizers and logical operators are updated using the following update rules, which will be used in the distance algorithm in Subsection \ref{subsec: algorithm}. These update rules are also basis independent, a property that will be relied on in proofs in Subsection \ref{subsec: proof}. 

\begin{lemma}(\textbf{Stabilizer Update Rules})\label{lemma: stab update rule}

    Let $\mathcal{S}$ be the stabilizer generators with a stabilizer state $\ket{\psi}$ being either in $+1$ or $-1$ eigenstate of the generators. 

Let $m$ be a Pauli measurement performed on $\ket{\psi}$, and denote the outcome of $m$ by $O(m)\in \{\pm 1\}$.

\begin{enumerate}
    \item If $\pm m \in \langle \mathcal{S} \rangle$, then the outcome is fixed by the eigenvalues of stabilizers for $\ket{\psi}$, and the state remains unchanged. 
    \item If $m$ anti-commutes with some elements in $\mathcal{S}$: Let $V = \{s_1, s_2,\cdots, s_l\}$ be a subset of $\mathcal{S}$ whose elements anti-commute with $m$. We replace $s_1$ with $m$ and update the rest of $V$ by $s_i \to s_i \cdot s_1$, for $2 \leq i\leq l$. $\mathcal{S} \cap V$ is now updated to $ \{O(m) \cdot m, s_2\cdot s_1, s_3 \cdot s_1,\cdots, s_l\cdot s_1\}$ 
    \item If $\pm m \notin \langle \mathcal{S} \rangle$ and $[m,s]=0 \, \forall s \in \mathcal{S}$, then we update the set of stabilizer generators: $\mathcal{S} \to \mathcal{S} \cup \{O(m) \cdot m\}$. This assumes that $m$ is not a logical operator. 
\end{enumerate}
\end{lemma}

\begin{lemma}(\textbf{Logical Update Rules})\label{lemma: logical update rule}

Let $L$ be a logical operator of a stabilizer group $\langle S\rangle $ and let $\ket{\psi}$ be an eigenstate of $L$.

Let $m$ be a Pauli measurement performed on $\ket{\psi}$ and denote the outcome by $O(m) \in \{\pm 1\}$.
\begin{enumerate}
    \item If $m = (-1)^a \cdot L$, then $O(m)\cdot (-1)^a$ gives the eigenvalue of $L$ for the state $\ket{\psi}$, and the logical operator remains unchanged. 
    \item If $m$ commutes with $L$, the logical operator remains unchanged. 
    \item If $m$ anti-commutes with $L$ and commutes with $\langle S\rangle $, then $L$ is updated to $O(m) \cdot m$. The new state is a $+1$ eigenstate of $O(m) \cdot m$ instead of $L$.
    \item If $m$ anti-commutes with $L$ and anti-commutes with some elements in $S$: In the stabilizer update rules, we replace an element $s_1$ with $m$ and update the rest of the elements in $S$ that anti-commute with $m$ using the $ 2^{\mathrm{nd}} $ rule in Lemma \ref{lemma: stab update rule}. For the logical operator, we update $L\to L\cdot s_1$, where $s_1$ is the element that is replaced with $m$.
\end{enumerate}
\end{lemma}

\subsection{Dynamical code and Subsystem code}
\begin{definition}
    A dynamical code is a code defined on a measurement sequence, that consists of rounds of measurements, with each round given by a set of commuting pauli measurements. After the $i^{\mathrm{th}}$ round of measurements, the code is defined by the stabilizer group that has been updated by the measurements according to the stabilizer update rules in Lemma \ref{lemma: stab update rule} defined above. The stabilizer group at this timestep is called the $i^{\mathrm{th}}$ instantaneous stabilizer group (ISG), $\langle \mathcal{S}_i\rangle$. 
    The codespace $\mathcal{C}_i$ is given by the ISG $\langle \mathcal{S}_i\rangle$ after the $i^{\mathrm{th}}$ round of measurements.

    \begin{equation}
        \mathcal{C}_i = \{ \ket{\psi} =\pm s \cdot \ket{\psi}, \forall s \in \mathcal{S}_i\}
    \end{equation}
    where the signs depend on the measurement outcomes for $s \in \mathcal{S}_i$.
    
\end{definition}

Between any two rounds of a dynamical code, we can define a subsystem code.  A subsystem code can be thought of as a stabilizer code with quantum information encoded into only a subset of the logical qubits. The logical qubits that are not used are called gauge qubits. Generally, we also do not correct for errors on the gauge qubits, so any logical operator should have a tensor product structure over $\mathcal{H}_{\mathrm{logical}}\otimes \mathcal{H}_{\mathrm{gauge}}$. However, it is also possible to have logical unitaries that do not respect this tensor product structure, then the logical operator will be dependent on the state of gauge qubits. One example is converting from a Steane code to a Reed Muller code for implementing transversal T gate. In this case, the non-Clifford logical operator depends on the state of the gauge qubits before performing the T gate.

\begin{definition}
  A subsystem code is defined by its gauge group $\mathcal{G}$. The stabilizers are given by $(Z(\mathcal{G})\cap \mathcal{G})/\mathbb{C}$. $\mathcal{G} = \langle S\rangle $ for stabilizer codes, with the centralizer of $\mathcal{G}$, $Z(\mathcal{G})$, defined as $Z(\mathcal{G}) = \{p \in \mathcal{P}_n: \forall g \in \mathcal{G}, gp =pg\}$.
\end{definition}

 \begin{definition}
     Given a set of independent stabilizer generators $S$, a destabilizer $\kappa \in K$ of a stabilizer generator $s \in S$ is a Pauli operator in $\mathcal{P}_n/\mathbb{C}$ that only anti-commutes with $s$, but commutes with $S\backslash s$ and the bare logical operators of the gauge group $\mathcal{G}=\langle S \cup K \rangle$.
 \end{definition}
 
When treating neighboring rounds as subsystem codes, we will need to discuss both the bare and dressed logical operators, whose definitions we list below: 

\begin{definition}
   Bare logical operators are logical operators that commute with all elements of the gauge group: $L_{\mathrm{bare}} = Z(\mathcal{G})$. They only act on the logical qubits but act trivially as identity on the gauge qubits.  
\end{definition}

\begin{definition}
   Dressed logical operators are logical operators that commute with the stabilizers: $L_{\mathrm{dressed}} = Z(S)$. They act on the logical qubits but may also act non-trivially on the gauge qubits. Thus, they may anti-commute with some of the elements in $\mathcal{G}$. Bare logical operators are a subset of dressed logical operators. When the code is a stabilizer code, the bare and dressed logical operators are the same. 
\end{definition}

We can define the distance for both stabilizer and subsystem code as the minimum weight of all non-trivial dressed logical operators.

\begin{definition}
    \begin{equation}
        d_{\mathrm{subsystem}}= \min \mathrm{wt}\{L_{\mathrm{dressed}}\backslash \mathcal{G}\}
    \end{equation}
\end{definition}

We can update the stabilizer group efficiently if we are performing only Clifford gates and Pauli measurements according to the Gottesman–Knill Theorem. This forms the foundation in the distance algorithm to update the stabilizer groups.

\section{Dynamical code}
\label{sec: dynamical code}

\subsection{Definition of Distance for a Dynamical Code}\label{subsec: distance}

\begin{figure}[ht]
    \centering
    \includegraphics[width=0.5\linewidth]{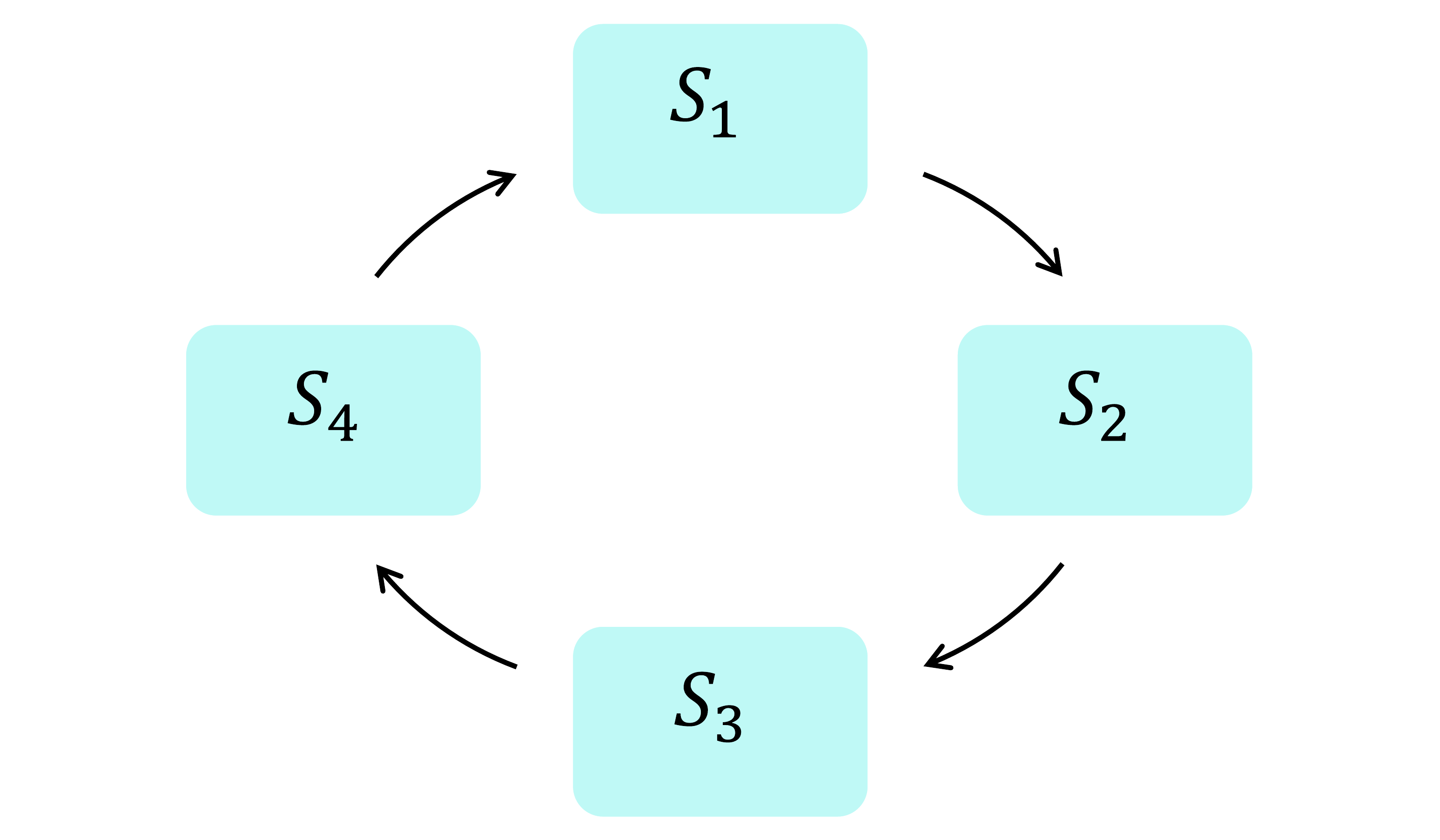} 
    
    \caption{The figure illustrates an example of a Floquet code which is a special case of dynamical codes consisting of periodic measurements (In this case, 4 rounds of measurements). Each round of measurements consists of commuting Paulis. The $i^{\mathrm{th}}$ round of measurements brings the code from the ISG defined by $S_{i-1}$ to the next ISG defined by $S_{i(\mathrm{mod} \: 4)}$.} 
    \label{fig: cyclical} 
\end{figure}

Unlike a static code with a fixed set of stabilizers or gauges, a dynamical code is defined by a sequence of measurement rounds with each round consisting of commuting Paulis. The stabilizer state evolves after each round of measurements with new representatives of logical operators. The distance of a dynamical code is upper bounded by the ISGs $S_i$, which define the codes $\mathcal{C}_i$ after $i^{\mathrm{th}}$ round of measurements (Refer to Fig \ref{fig: cyclical} for illustration). Essentially, for a dynamical code to exhibit a good distance as a whole, each $S_i$ must minimally have a good distance. Otherwise the encoded logical information will lose its protection from physical errors in the ISG with a small distance. 

We begin by outlining the problem under consideration. For a given dynamical code, what is the set of errors that an ISG can correct given the syndrome information acquired through subsequent rounds of measurements? Using this information, what is the distance of the ISG? Without loss of generality, we examine the ISG at round 0 of the code and assume that all syndromes for this ISG come from the measurement sequence. 

When defining the distance, we assume that the errors only occur at round 0 of the code and we also assume that if the errors are trivial, we know what the syndromes for the stabilizers in round 0 will be.  The assumption that all errors occur in round 0 means that this notion of distance does not fully capture fault tolerance of the dynamical code.  Instead, it is the analogue to the distance of a static quantum error-correcting code, which gives a necessary condition for fault tolerance, but not a sufficient one.  However, note that the distance algorithm has wider applicability: It works by classifying the generators of the ISG as unmasked, temporarily masked, and permanently masked, which makes no reference to the actual errors in the system.  Only the definition of distance itself uses the assumption that the errors are at round 0; the classification of stabilizer generators produced by the algorithm is still correct when errors can occur at any point.

One crucial reason for defining the distance of a dynamical code is that even though an ISG can be treated just like a stabilizer code with all the syndromes given by a complete set of stabilizer generators for the ISG, not all the syndromes can be obtained in a dynamical code. Thus, the dynamical code may not be able to realize the full error correcting capacity that is captured by the distance of the ISGs, so we cannot use the same notion of distance for stabilizer codes to define the distance of a dynamical code. Therefore, it is important to reexamine the concept of distance for that of a dynamical code.

Another possible notion of distance for dynamical codes that one can think of is by treating them as generalized subsystem codes in the sense that every neighboring pair of ISGs can be seen as a subsystem code. We can find a set of stabilizer generators for $S_i$ and $S_{i+1}$, such that there are pairs of anti-commuting stabilizers $s_\alpha \in S_i$, $s_{\beta}\in S_{i+1}$ satisfying $\{s_\alpha,s_{\beta}\}=0$ while all other pairs commute. Then these pairs give the gauge qubits. The outcomes from measuring these gauges are random, so a better formulation will be to treat the code as a subsystem code. Thus, a better upper bound for the distance of a dynamical code is given by the minimum of $d_{\mathrm{subsystem}}$ formed by each subsystem code from two neighboring rounds of ISG. 

In this work, we found an even tighter notion of distance for dynamical code using the fact that the code is defined to have error correction properties given by its sequence of measurements. This upper bound for distance is tighter and is given for a particular round by the syndrome information that can be obtained through measurements from subsequent rounds. Taking the minimum across all rounds gives a tighter upper bound for the dynamical code's distance. 

To this end, we introduce the notion of unmasked distance of the code, where we allow for the possibility that some of the syndromes cannot be obtained. The stabilizers corresponding to the unknown syndromes are called masked stabilizers.

We present an algorithm for finding the masked stabilizers and calculating the unmasked distance for round 0. Since distance finding is naturally a NP hard problem, it remains so for dynamical codes, unless a more tailored algorithm can be found for a more specific family of codes. The part of the algorithm that determines the masked and unmasked stabilizers is classically efficient.

To summarize, the various notions of distances for dynamical code are defined below. 

\begin{itemize}
    \item Distance $d_{\mathrm{ISG}}$ is given by the minimum of the distances of all ISGs. 
 
 \item Distance $d_{\mathrm{subsystem}}$ is given by the minimum of the distances of all subsystem codes formed by neighboring ISGs.
 
 \item Distance $d_{\mathrm{u}}$ is given by the minimum of the unmasked distances of all ISGs.
\end{itemize}

Note that the actual distance of the code is upper bounded by $d_{\mathrm{u}}$ with $d_{\mathrm{u}}\leq d_{\mathrm{subsystem}}\leq d_{\mathrm{ISG}}$. 

\begin{example}
    In the Bacon Shor code \cite{Bacon_2006}, the bare logical operators are single column or row of Pauli X or Z, but the dressed logical can be up to products with any gauges. The minimum weight of both dressed and bare logical operators are the same, so $d_{\mathrm{subsystem}} = d_{\mathrm{ISG}}$. 

\end{example}
\begin{example}
Another example is the subsystem surface codes with three-qubit check operators in \cite{bravyi2013subsystem}. In this case, $ d_{\mathrm{subsystem}}\leq d_{\mathrm{ISG}}$. If the gauges are fixed in the x basis, then the distance for a $[[3L^2,2,L]]$ code is given by $L$ for $X$ logical and $2L$ for $Z$ logical, but for $d_{\mathrm{subsystem}}$, the minimum weight of dressed logical is $L$ for both $X$ and $Z$ logical operators. 
    
\end{example}
\begin{example}
    In the Hastings and Haah’s honeycomb code \cite{Hastings2021dynamically}, the weight 2 checks are the masked stabilizers for each ISG. Both inner and outer logical operators have the same weight as dressed logicals in a subsystem code for neighboring rounds of an ISG or as bare logicals in an ISG. Thus, the code satisfies $d_{\mathrm{subsystem}}\leq d_{\mathrm{ISG}}$.
\end{example}

\subsection{Masking and Unmasked Distance of a Dynamical Code}

Given a sequence of measurements for a stabilizer group $\langle S\rangle $, we can classify the generators in $S$ as one of the following three types by how accessible their syndrome information is using this sequence. 

\begin{definition}
    An unmasked stabilizer is a stabilizer whose outcome can be obtained through measurements in the sequence. 
\end{definition}

\begin{definition}
     A temporarily masked stabilizer is a stabilizer whose syndrome cannot be obtained given the sequence of measurements. However, the syndrome information has not been irreversibly lost and can potentially be extracted by creating new measurements. 
 \end{definition}

\begin{definition}
    A permanently masked stabilizer is a stabilizer whose outcome can never be obtained through the sequence of measurements or even by adding new rounds of measurements. 
\end{definition}

\begin{figure}[t]
    \centering
    \includegraphics[width=0.60 \linewidth]{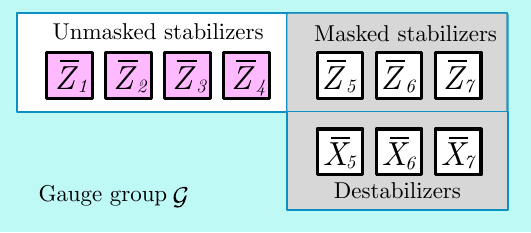} 
     
    \caption{The figure illustrates an example of a code with 7 stabilizers, out of which 4 are unmasked stabilizers (pink online) and 3 are masked stabilizers. The destabilizers are chosen to anti-commute with their respective masked stabilizers. The entire set (blue online) forms the generators of a gauge group $\mathcal{G}$. The distance of this subsystem code is given by $d_{\mathrm{subsystem}} = \min\:\mathrm{wt\:} \{\mathcal{N}(U)\backslash \mathcal{G}\}$.}
    \label{fig:gauge group}
\end{figure}

In this paper, we let $U$ denote the set of unmasked stabilizers, $T$ denote the set of temporarily masked stabilizers, $P$ denote the set of permanently masked stabilizers and $S$ the set of all stabilizer generators for a particular ISG in the measurement sequence. 

\begin{remark}
	Note that $u\cdot p \in P$ for any $u \in U, p \in P$ and $u \cdot t \in T$ for any $u \in U, t \in T$. $p \cdot t \in P$ for any $p \in P$ and $t \in T$.
\end{remark}    

\begin{remark}
	One can extend the definition of $T$ to that of a group: Define $\tilde{T} = \langle T\rangle U$. Then, $\tilde{T}$ consists of the group of unmasked and temporarily masked stabilizers. We have $U \subseteq \tilde{T}$ and $P= \langle S\rangle \backslash \tilde{T}$.
\end{remark}
    
We now show that the distance of a masked stabilizer code is given by that of a subsystem code with gauge group $\mathcal{G} $ that is in part determined by the measurement sequence. Let $W$ be the set of errors that are correctable. Errors $E, F\in W$ have to satisfy either $E^{\dagger}F \in S$ or that they give distinct unmasked syndromes, $E^{\dagger}F \notin \mathcal{N}(U)$, in order for $W$ to be correctable, where $U$ is the set of unmasked stabilizer generators defined earlier and $\mathcal{N}(U)$ is the normalizer of $\langle U\rangle $. Thus, any pair of errors $E$ and $F$ in $W$ satisfies the following condition:
\begin{equation}\label{incorrect_distance}
    E^{\dagger}F \in (\mathcal{P}_{\mathrm{n}} -\mathcal{N}(U)) \cup \langle S \rangle /\mathbb{C}
\end{equation}

The distance given Eq. \eqref{incorrect_distance} is then defined as according to \cite{gottesman2022opportunities}:

\begin{equation}
    d_{\mathrm{s}} = \min\:\mathrm{wt\:} \{\mathcal{N}(U)\backslash \langle S \rangle \}
\end{equation}

However, there are cases where even though two errors differ on the syndrome for masked generators, they do not affect our measurements of logical states. Let's assume $E^{\dagger} F \in \mathcal{N}(U) \backslash \mathcal{N}(S)$, so $E$ and $F$ have the same unmasked syndrome but different masked syndromes. If $L$ and $L'$ are two representatives for the same logical qubit, then measuring either is normally a valid way of measuring the logical value. Now, if $[E^{\dagger} F, L] = 0$, then it does not matter whether we have an error $E$ or an error $F$; in either case, we will get the same eigenvalue for $L$ when we measure it. On the other hand, if $\{E^{\dagger} F, L' \} = 0$, then the error $E$ followed by measuring $L'$ will give us a different result than the error $F$ followed by measuring $L'$. In other words, we can pick a coset of logical representatives for which $[E^{\dagger} F, L] = 0$ and this will allow us to correctly deduce the logical outcomes despite a larger set of errors that now consists of indistinguishable errors like $E$ and $F$.

If there are $l$ masked stabilizer generators, then we have $2^l$ different ways of choosing cosets of logical representatives for a given logical operator. Each choice is equivalent to picking a unique set of destabilizers for the masked stabilizers. The set of destabilizers, the masked stabilizers and the unmasked stabilizers together form the generators of the gauge group $\mathcal{G}$ (See Fig \ref{fig:gauge group}). The logical operators $L$ are chosen to commute with all elements of $\mathcal{G}$. 

\begin{definition}
    The set of allowable errors $V$ for a stabilizer group with masking satisfies the following condition:     \begin{equation}\label{correctable}
    E^{\dagger}F \in [(\mathcal{P}_{\mathrm{n}}-\mathcal{N}(U))\cup \langle S\rangle ]\cdot \mathcal{G}
    \end{equation}
    with $E,F \in V$, and the gauge group $\mathcal{G}$ defined by the unmasked stabilizers, the masked stabilizers and their destabilizers.
    
\end{definition}
 This set includes errors that can give non-trivial syndromes on the masked stabilizers too, but they are correctable as they do not give a logical error for our choice of logical representatives.

\begin{definition}
The distance is defined as the minimum weight of $E^{\dagger}F$ that does not satisfy Eq. \eqref{correctable}, which is the complement of the set $[(P_{\mathrm{n}}-\mathcal{N}(U))\cup \langle S\rangle ]\cdot \mathcal{G}$. This is exactly given by 
    \begin{equation}
    d_{\mathrm{u}_2} = \min\:\mathrm{wt\:} \{\mathcal{N}(U)\backslash \mathcal{G}\}
\end{equation}

This is also the distance of the minimum weight of the dressed logical operators, by treating the masked stabilizers and their destabilizers as gauge qubits. $d_{\mathrm{u}_2}$ is the same as the distance for subsystem code defined on $\mathcal{G}$. Interpreted this way, $L$ is by definition the bare logical operators.
\end{definition}

With the fixed choices of destabilizers, we have completely determined the set of bare logical operators, up to $U$. We see that there are 3 equivalent notions here: Choosing a gauge group $\mathcal{G}$ is the same as fixing the set of destabilizers which is also equivalent to choosing the bare logical representatives. We can see that doing so allows us to correct for more errors, since it also allows for errors that are destabilizers of the masked stabilizers, because these errors satisfy the condition that $[E^{\dagger}F, L]=0$, for all bare logical representatives $L$s of the new gauge group.

\begin{lemma}
    If the stabilizer syndrome is that of a permanently masked stabilizer, then the choice of destabilizer has been fixed by the measurement sequence.
\end{lemma}

\begin{proof}
    
Suppose we have 2 rounds of ISGs given by $\{S_1, S_2\}$. In the first round, the logical operators are $L_x$ and $L_z$ and one of the stabilizers is labelled $s_1$. In the second round, $s_2$ is measured and $s_2$ anti-commutes with only $s_1$. This means that $s_1$ is a permanently masked stabilizer, and the code forms the usual subsystem code, with $\mathcal{G}=\langle S_1 \rangle\cup \{s_2\}$, that we are familiar with. The bare logical operators are those that commute with both $s_2$ and $s_1$ and belong to $N(\mathcal{G})\backslash \mathcal{G}$. The gauge qubit formed here is the one represented by $s_1$ and $s_2$ and the Hilbert space can be decomposed as $\mathcal{H}_{\mathrm{gauge}}\otimes \mathcal{H}_{\mathrm{logical}}$.

For a code with more than two rounds, the above scenario is true in any neighboring rounds of ISGs. The anti-commuting pairs of stabilizers from neighboring rounds of ISGs define the gauge qubits. We can treat them as logical qubits that we did not use to encode logical information and define bare logical operators that commute with the logical operators for the gauge qubits. Since the stabilizer update rules uniquely determine how the logical operators evolve with measurements, we can work backwards to retrieve the bare logical operators for each round, till we reach the ISG in round 0 of the dynamical code. An algorithm that does this is algorithm \ref{alg: pms and destabilizer} given in Subsection \ref{subsec: algorithm}.

\end{proof}

On the other hand, if it is a temporarily masked stabilizer, we can optimize the unmasked distance over all possible choices of the destabilizers for this stabilizer. 

This protocol gives us the most optimal choice of destabilizers, which together with $S$, defines a gauge group $\mathcal{G}$. From here, we can obtain the unmasked distance of a dynamical code. We summarize our discussion below as a key definition for the unmasked distance of an ISG given a particular sequence of measurements. 

\begin{definition} \label{def:unmasked distance}
For a masked stabilizer code with $l$ masked stabilizers, its distance is given by:
\begin{equation}
    d_{\mathrm{u}} \coloneqq \min\:\mathrm{wt}\{ \mathcal{N}(U)\backslash \mathcal{G}\}
\end{equation}
  
Here $\mathcal{G}$ depends partly on the freedom in the choice of destabilizers for the temporarily masked stabilizers and partly on the measurement sequence which fixes the destabilizers for the permanently masked stabilizers.

\end{definition}

We conclude this subsection with an example to illustrate $d_{\mathrm{u}}$, by considering a stabilizer code with masking.

\begin{example}
    As a simple illustration, the Shor code is a [[9,1,3]] code with logical operators given below.
    \begin{eqnarray*}
        z_{\mathrm{L}} &=& x_1x_2x_3\\
        x_{\mathrm{L}} &=& z_1z_4z_7 
    \end{eqnarray*}
    And the following 8 stabilizers:
    \begin{eqnarray*}
      &&z_1 z_2, z_2 z_3, z_4 z_5,z_5 z_6, z_7 z_8, z_8 z_9\\ &&x_1 x_2 x_3 x_4 x_5 x_6, x_4 x_5 x_6 x_7 x_8 x_9
    \end{eqnarray*}
    If we mask the stabilizer generator $z_1z_2$, then the error $x_1$ and the trivial error are now indistinguishable.  We can define the unmasked distance as a subsystem code with an extra gauge operator $x_1$. This choice of gauge fixes the bare logical $x$ and $z$ of the code to be $x_1 x_2 x_3$ and $z_2 z_4 z_7$. Thus, the unmasked distance is $2$, as the smallest non-trivial dressed logical is $x_2 x_3$. Suppose we choose a different gauge $x_2 x_3$; then the bare logical operators change to $x_1 x_2 x_3$ and $z_1 z_4 z_7$. The unmasked distance becomes $1$. 
\end{example}

\subsection{Logical Operators of a Dynamical Code} \label{subsec: Logical operators}

\begin{figure}[ht]
    \centering
    \includegraphics[width= \linewidth]{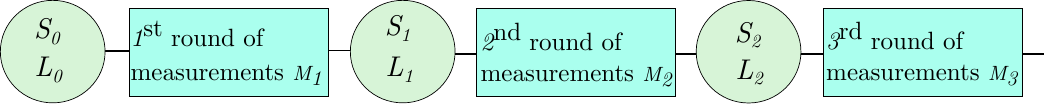} 
     
    \caption{The figure illustrates the first three rounds of a sequence of measurements. $S_i$ is the set of stabilizer generators of the ISG in the $i^{\mathrm{th}}$ round and $L_i$ is the updated logical operator from a given logical representative $L_0$ in the $0^{\mathrm{th}}$ round.}
    \label{fig: logical update}
    
\end{figure}
To recover the encoded logical information in a dynamical code, we need to take into account both the errors and the measurement outcomes.

In the absence of Pauli errors, the logical outcome can be obtained by the logical update rules. Here, we generalize the computation of logical outcome to include Pauli errors in all rounds of a dynamical code. It can be shown that the logical outcome depends on the temporal order of measurements and errors. This suggests that in general protecting the encoded state can be more challenging for a dynamical code than a stabilizer code. The scenario that we are interested in where errors only occur in the first ISG of the dynamical code for the setup described in Section \ref{sec: the distance algorithm} is included as a special case.

\begin{theorem}\label{theorem: logicals}
Suppose we initialize a dynamical code that consists of $l$ rounds of measurements in an eigenstate of $L_0$, which is fixed as our choice of logical representative for the $0^{\mathrm{th}}$ ISG, with the notations illustrated in Figure \ref{fig: logical update}. Each round of measurements consists of commuting Pauli operators which form the set $M_i$ for the $i^{\mathrm{th}}$ round. Let $E_k$ be the net error from the time after the $k^{\mathrm{th}}$ round of measurements till the final round of the code. Then, $E_{0}$ is the net error since initialization till the final round of the code. The Pauli operators $L_{0}$ and $L_{l}$ are related by $L_{l} = L_{0}\cdot s_0\cdot \prod_{i=1}^{l-1} m_i$ up to $\pm 1$ sign difference. $m_i$ is an element of $\langle M_i \rangle$ and $O(m_i)$ is the outcome of $m_i$ which can be obtained through the outcomes of elements in $M_i$. $s_0$ is an element of $\langle S_0 \rangle$ which is the stabilizer group for round $0$ of the dynamical code and $O(s_0)$ is the outcome of $s_0$. $A \odot B$ is the binary symplectic product of A with B as defined in Subsection \ref{subsec: outcome and syndrome}. 

The outcome of the logical $L_l$ is given by the following equation. 

\begin{equation}\label{eq: logicals of dynamical code}
O(L_{l}) = O(L_{0}) \cdot (-1)^{L_{0} \odot E_{0}} \cdot \left(\prod_{i=1}^{l-1} O(m_i) \right)\cdot O(s_0) \cdot (-1)^{E_0 \odot s_0}\cdot (-1)^{\sum_{k = 1}^{l-1} E_k \odot m_k}
\end{equation}

\end{theorem}

The case where errors only occurred in round $0$ follows as a special case. 

\begin{corollary}
The outcome of the logical if errors only occurred in round 0 of the code is given by: 
    \begin{equation}\label{eq: logicals of dynamical code simplified}
O(L_{l}) = O(L_{0}) \cdot (-1)^{L_{0} \odot E} \cdot \left(\prod_{i=1}^{l-1} O(m_i) \right) \cdot O(s_0) \cdot (-1)^{ E \odot s_0},
    \end{equation}
where $E$ is the Pauli error that occurred in round 0 of the dynamical code. 

\end{corollary}

\begin{proof}[Proof of Theorem \ref{theorem: logicals}]
    
    The outcome of a logical representative changes in two ways. First, because of measurements, the Pauli operator for the logical representative is updated by some stabilizers in the current ISG. Secondly, after the measurements, some errors can happen on the physical qubits and this can change the logical outcome. We explicitly explain how both steps can change the logical outcome. 

    During the $i^{\mathrm{th}}$ round of measurements, the logical representative is updated according to the logical update rules in Lemma \ref{lemma: logical update rule}. In the case where we make a measurement that anti-commutes with both the logical representative and a stabilizer generator $s$ in the current ISG, the logical is updated from $L \to L\cdot s$ then the outcome of the new logical is now given by $O(L) \cdot O(s)$. $s$ is a stabilizer in the current round, and it can be expressed as a product of stabilizers from the $0^{\mathrm{th}}$ ISG and measurements prior to the $i^{\mathrm{th}}$ round. The outcome $O(s)$ is the outcome of the stabilizer $s$ if it is measured at the current time step, which can be different from the initial measurement outcomes because of errors that could have occurred in between. 
    
    After the round of measurements, some errors could have occurred on the code. If the errors anti-commute with the updated logical representative, then the outcome for the logical operator will acquire a sign flip from the errors. 

    By applying these updates for all $l$ rounds of the dynamical code, the outcome for the logical operator in the last round can be summarized in a single equation as follows: 
    
    \begin{equation}\label{eq: logicals rearranged}
    O(L_{l})= O(L_{0})\cdot (-1)^{L_{0} \odot E\mid_0} \prod_{i = 0}^{l-1}O(s_i)\mid_{i+1} \prod_{i = 1 }^{l} (-1)^{L_{i} \odot E\mid_{i}},
    \end{equation}
    
    where $s_i \in \langle S_i \rangle$ is a stabilizer element from $i ^{\mathrm{th}}$ ISG that is multiplied to the logical operator through $(i+1)^{\mathrm{th}}$ round of measurements, and $O(s_i)\mid_{i+1} $ is the outcome of $s_i$ immediately prior to the $(i+1)^{\mathrm{th}}$ round of measurements. $L_{i}$ is the logical operator from the $i^{\mathrm{th}}$ ISG and $E\mid_{i}$ is the error that occurred in the  $i^{\mathrm{th}}$ ISG.

    From here, we note that $s_i \in \langle S_i\rangle$ and $s_i$ is a stabilizer element that has been updated by the stabilizer update rules, that is $s_i$ is the product of some elements $s_0 \in \langle S_0\rangle$ and measurements from the previous rounds. Thus, we can reexpress each $s_i$ as $s_i = s_0 \cdot \prod_{j = 1}^{i} m_j$ for some $s_0 \in \langle S_0\rangle$ and some $m_j \in \langle M_j \rangle$. Note that $m_j$ need not be an actual measurement, but it is in the group generated by the measurements from the $j^{\mathrm{th}}$ round of measurements. After replacing the stabilizers in Eq. \eqref{eq: logicals rearranged} with this, we can rearrange the elements to obtain Eq.\eqref{eq: logicals of dynamical code}.

\end{proof}
    We give an example to illustrate that we only need to keep track of the commutation relation between a measurement made at time $t$ and the errors that happened after time $t$. 
\begin{example}
   	We can construct an example where an error $e$ commutes with a logical operator $L_1$ in an earlier round but only anti-commutes with the updated logical operator, $L_1 \cdot s_2$, in a later round. Because the measurement $s_2$ enters the stabilizer group at a later round than when the error occurred, in this case, the logical outcome is unaffected by the error that has occurred on the code. This can be counter-intuitive since the error $ e $ anti-commutes with the updated logical.

    Suppose in round 1, we have a logical operator $L_1$, stabilizer $s_1$ in $S$ and error $e$ that occurred before round 1. Let $e\ket{\psi}$ be the stabilizer state in round 1 after $e$ occurs. In round 2, the logical operator remains as $L_1$, $s_2$ is a new stabilizer that anti-commutes with $e$ but commutes with $s_1$ and in round 3, we measure $s_3$ that anti-commutes with both $L_1$ and $s_2$ so that the new logical in rounds 3 is now updated to $L_1 \cdot s_2$ and we have $s_3$ in the stabilizer group. This is illustrated in the table below:
    
\begin{center}
    \begin{tabular}{ |p{2cm}|p{2cm}|p{2cm}|p{2cm}|  }
 \hline
 ISG & Stabilizers & Logical & Events\\
 \hline 
    $1$ & $s_1,\cdots$ & $L_1$ & \cellcolor[HTML]{d66793} Error $e$\\
 \hline
 \multicolumn{4}{|l|}{Measurement: Measure $s_2$, $\{e,s_2\}=0$} \\
 \hline
 $2$ & $s_1,s_2,\cdots$ & $L_1$ & \\
 \hline
    \multicolumn{4}{|l|}{Measurement: Measure $s_3$, $\{s_3,s_2\}=0, \{s_3,L_1\}=0$} \\
\hline
$3$& $s_1, s_3,\cdots $&$L_1 \cdot s_2$ &\\
 \hline
\end{tabular}
\end{center}

Assuming $[L_1,e]=0$, even though $L_1\cdot s_2$ anti-commutes with $e$, the correct outcome of the logical operator is just given by $= O(L_1) \cdot O(s_2)$, independent of the error $e$ that happened earlier in the measurement schedule. 

\begin{align*}
&(L_1\cdot s_2) \left(1\pm s_3\right)\left(1\pm s_2\right)e \ket{\psi}\\
= &\left(1\pm s_3\right) (L_1\cdot s_2)\left(1\pm s_2\right)e \ket{\psi}\\
= &\left(1\pm s_3\right) (\pm 1)\left(1\pm s_2\right)L_1 \cdot e \ket{\psi}\\
= & O(s_2)\cdot O(L_1) \left(1\pm s_3\right)\left(1\pm s_2\right)e \ket{\psi}\: (\mathrm{assuming\:}[L_1,e]=0)
\end{align*}
    
\end{example}

In the equation above, we omitted the normalization factor, and projected the encoded state onto the eigenspace of $s_2$ and $s_3$, with the $\pm$ signs dependent on the outcomes of $s_2$ and $s_3$, i.e. if $O(s_2)= 1$, then the projector is $(1+s_2)$.

This shows that the logical outcome depends on the order of occurrence between the measurements and the errors. The key point is that when an error anti-commutes with a measurement but the error occurred first, the outcome of the measurement is unaffected by the error. Thus the logical outcome in round 3 is only given by $O(L) \cdot O(s_2)$, even though $s_2$ anti-commutes with $e$.

\begin{example}
On the other hand, if the error $e$ occurred in round 2, after measuring $s_2$, we see a change in the logical outcome. Here, we again assume that $ [L_1, e] = 0 $

\begin{center}
    \begin{tabular}{ |p{2cm}|p{2cm}|p{2cm}|p{2cm}|  }
 \hline
 ISG & Stabilizers & Logical & Events\\
 \hline 
    $1$ & $s_1,\cdots$ & $L_1$ & \\
 \hline
 \multicolumn{4}{|l|}{Measurement: Measure $s_2$, $\{e,s_2\}=0$} \\
 \hline
 $2$ & $s_1,s_2,\cdots$ & $L_1$ &\cellcolor[HTML]{d66793} Error $e$\\
 \hline
    \multicolumn{4}{|l|}{Measurement: Measure $s_3$, $\{s_3,s_2\}=0, \{s_3,L_1\}=0$} \\
\hline
$3$& $s_1, s_3,\cdots$&$L_1 \cdot s_2$ &\\
 \hline
\end{tabular}
\end{center}
\begin{align*}
&(L_1\cdot s_2) \left(1\pm s_3\right)e \left(1\pm s_2\right)\ket{\psi}\\
= &\left(1\pm s_3\right) e (-1)(L_1\cdot s_2)\left(1\pm s_2\right)\ket{\psi}\\
= &\left(1\pm s_3\right) e (-1)(\pm 1)\left(1\pm s_2\right)L_1 \ket{\psi}\\
= & -1 \cdot O(s_2)\cdot O(L_1) \left(1\pm s_3\right) e\left(1\pm s_2\right) \ket{\psi}\: (\mathrm{assuming\:}[L_1,e]=0)
\end{align*}
    
\end{example}

We also use Floquet code as an example to show how a logical representative's outcome is obtained. 

\begin{example}[Hastings and Haah honeycomb code]
    For a logical representative of an outer logical operator in the honeycomb code, the logical operator is updated by some weight 2 checks in each measurement round. The logical outcome in the $(i+1)^{\mathrm{th}}$ round is given by the product of the outcome of the logical operator in the $i^{\mathrm{th}}$ round and the outcomes of the weight 2 checks in the $i^{\mathrm{th}}$ round that are multiplied to the logical operator. Further, if the errors that occurred after the $(i+1)^{\mathrm{th}}$ round of measurements anti-commute with the new logical representative, then the outcome of the logical representative will flip sign.  
\end{example}

In general, we must keep track of the space-time positions of errors to recover the logical information.

Specifically, we need information of the following:
\begin{enumerate}
    \item Whether errors that occurred in round $i$ commute with the logical operators in round $i$. 
    \item The outcomes of stabilizers immediately prior to being multiplied to the logical operators. 
\end{enumerate}

\subsection{Syndromes through Measurements in a Dynamical Code}\label{subsec: syndromes through}

Subsection \ref{subsec: Logical operators} shows the importance of deducing the time that errors occur, in order to correctly compute the logical outcome. In this subsection, we show that the temporal positions of errors also affect the syndromes of the errors. 

In a static stabilizer code, the syndrome for an error only depends on the commutation relations between the stabilizers and the errors. In a dynamical code, if we allow for errors to occur in all rounds of the dynamical code, the syndromes for an error also depend on the temporal position of the error. 

\begin{theorem}\label{theorem: spacetime errors}
    Let $s \in S_0$ be a stabilizer whose outcome is obtained through subsequent rounds of measurements, $s = \prod_{i = 1}^{f} m_i$ where $m_i \in \langle M_i \rangle$ and $\langle M_i \rangle$ is the group generated by the measurements $m_i$ in the $i^{\mathrm{th}}$ round of measurements. Let $E =\prod_{i = 0}^{f} e_i$ be the error that occurred on a dynamical code, where $e_i$ is the error that occurred in the $i^{\mathrm{th}}$ ISG.

    The syndrome of $ E $ for $s$ is $(-1)^a$, with $a$ given by:

    \begin{eqnarray}
        &&e_0 \odot \prod_{i=1}^f m_i + e_1 \odot \prod_{i=2}^f m_i + \cdots+ e_{f-1} \odot m_f \mod \: 2\\
        &=& \prod_j \left(e_j \odot \prod_{i=j+1}^f m_i\right) \mod \: 2
    \end{eqnarray}
    
\end{theorem}
\begin{proof}

    Consider the error $e_i$ in round $i$ of the ISG. The syndrome for $e_i$ is obtained by the outcome of $s$. We assume that $\{e_i,s\}=0$. If $e_i$ occurs before $s$ is measured, it has no effect on the outcome of $s$. If $e_i$ occurs somewhere in the middle of these measurements $m_i$, then the outcome of $s$ changes sign only if $e_i$ commutes with an even number of $m_j$'s, for $1\leq j\leq i$, that are measured before $e_i$ occurs, that is if $e_i$ anti-commutes with $\prod_{j=i+1}^f m_j$
         
\end{proof}

As the decoding problem is much harder for the situation where errors occur in all rounds of the dynamical code and can depend on the details of the dynamical code, we consider only the simpler scenario with errors happening only in round $0$ of the dynamical code, so that the outcomes for stabilizers in round 0 correspond only to syndromes of errors in round 0. From here, we can obtain an upper bound in the distance of the code. 

\section{The Distance Algorithm}\label{sec: the distance algorithm}

We consider a dynamical code $\mathcal{Q}$ with $l$ rounds of measurements. The ISGs are denoted by $S_0, S_1, \cdots, S_l$ starting with the ISG $S_0$ in round $0$. The sets of commuting measurements in each round are denoted by $M_1, M_2,\cdots$, where $M_i$ is the set of measurements made after the $(i-1)^{\mathrm{th}}$ ISG. 

The $l$ rounds of measurements can also be interpreted as the window of measurements for obtaining the syndrome information to decode for errors in round $0$. We assume that some errors occurred in round $0$ and no more errors happen in the later rounds of measurements and we only use $l$ rounds of measurements as our measurement window for learning the syndrome bits in round $0$. To gain a better understanding of the distance of the ISG in round $0$, we need to know for each syndrome bit in round $0$, which of the following 3 cases does it fall under given only the first $l$ rounds of measurements:
\begin{enumerate}
    \item The syndrome bit can be learnt using the given measurement sequence. Stabilizers for these syndrome bits are unmasked and we denote this set by $U$. 
    \item The syndrome bit will be erased and never learnt given the measurement sequence. Stabilizers for these are permanently masked, and this set is denoted by $P$. 
    \item The syndrome bit will not be erased but we cannot learn it in the $l$ measurement rounds. Stabilizers for these are temporarily masked and can potentially be unmasked if we increase the measurement window with more measurement rounds and this set is denoted by $T$.
\end{enumerate}

To understand why this classification may be difficult, we give some intuition in Subsection \ref{subsec:intuition} through examples that illustrate that whether a syndrome is unmasked or not can depend on both the scheduling and the choice of the Pauli measurements. This motivates the need for an efficient classical algorithm that can exhaustively search for all the unmasked, permanently masked and temporarily masked stabilizer generators, and the outcomes for the unmasked stabilizers, so that one can determine the upper bound for the distance of a dynamical code. We present this algorithm in Subsection \ref{subsec: algorithm}. Some examples are given in Subsection \ref{subsec:examples} and the proof is provided in Subsection \ref{subsec: proof}.

\subsection{The Intuition}
\label{subsec:intuition}

Measurement order is crucial in deciding if a stabilizer is unmasked or not. Here we use an ordered set so that the measurements are made in the order they appeared in the set. For example, $\{x_1x_2, x_3x_4\}$ means measuring $x_1x_2$ first followed by $x_3x_4$. 

\begin{example}

We consider a high weight stabilizer $s=x_1 x_2 x_3 x_4 x_5 x_6$, whose outcome is obtained using the following measurements $\{x_1x_2, x_3 x_4, x_5 x_6\}$ in the order they appear. However, suppose we insert a non-commuting Pauli in between: $\{x_1x_2,z_2z_3, x_3 x_4, x_5 x_6\}$. The second measurement anti-commutes with the first measurement but commutes with the stabilizer. In this case, we cannot obtain the outcome for the stabilizer. 

\end{example}

There are also scenarios where new measurements anti-commute with some of the previous measurements, yet the outcome from the previous measurements can still be used to obtain the outcome of the stabilizer. 

\begin{example}
    Suppose the order of measurements is given by $\{x_1x_2, x_3 x_4, z_2z_3, x_5 x_6\}$. It turns out that in this case the stabilizer outcome can be obtained. 
\end{example}

One can also construct examples where the stabilizer $s=x_1 x_2 x_3 x_4 x_5 x_6$ is no longer in the current ISG but its outcome can still be obtained at a later round. 

\begin{example}\label{ex: anti-commute}
    We consider $\{x_5x_6, z_6z_7, x_1x_2, x_3x_4\}$ as a measurement sequence. In this case even though $z_6z_7$ anti-commutes with both the stabilizer and $x_5x_6$, the outcome for the stabilizer $x_1 x_2 x_3 x_4 x_5 x_6$ can be obtained. 
\end{example}

Thus, we see interesting situations where even though the original stabilizers are no longer present in the code, we can still perform some measurements to obtain the outcome of the original stabilizers.

In addition, it is also possible to invent scenarios where two stabilizers end up looking "the same", but we can still obtain the outcomes for each of them. 

\begin{example}
    Consider a 4 qubit example with $S_1 = \{s^{(1)}_1 = x_1   x_2   z_3, s^{(2)}_1 = z_3   x_4, s^{(3)}_1= x_4, s^{(4)}_1 = x_1   x_2\}$. Suppose we measure $m_3 = x_3  z_4$ and $m_4 = z_1   x_3$. Then, we will obtain the following: $S_2 =\{s_2^{(1)}=s^{(1)}_1  s^{(3)}_1 = x_1  x_2  z_3   x_4, s_2^{(2)}=s^{(2)}_1   s^{(4)}_1 = x_1  x_2  z_3   x_4, s_2^{(3)} = x_3  z_4, s_2^{(4)} = z_1   x_3\}$.
\end{example} 

What is interesting here is that even though we ended up with the two stabilizers evolving to look exactly the same, it is still possible to use future measurements, such as $ x_1   x_2   z_3   x_4$ to obtain the outcomes for both stabilizers $s^{(1)}_1$ and $s^{(2)}_1$. 

In the above examples, it may seem somewhat obvious after some computation, but when presented with a series of measurements, it can be hard to tell what are all the generators that have known outcomes from a series of different measurements.

\subsection{Outline and Notation}
\label{subsec:notation}
\begin{figure}[ht]
    \centering
    \includegraphics[width=0.6\linewidth]{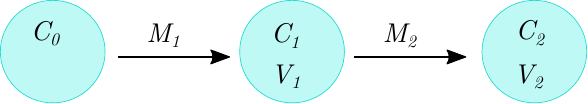} 
    
    \caption{An illustration of the first three rounds of a sequence of measurements that result in the updated sets of stabilizer generators. The union of $C_i$ and $V_i$ gives an overcomplete set of generators for $\langle S_i \rangle$ which is the ISG at round $i$ or after the $i^{\mathrm{th}}$ round of measurements. $M_i$ is the set of measurements made in the $i^{\mathrm{th}}$ round of measurements.}
    \label{fig:floquet code} 
\end{figure}

The goal is to find the unmasked distance of an ISG in a dynamical code given a time window of $l$ measurement rounds to extract its syndrome information. Without loss of generality, we set this to be the $0^{\mathrm{th}}$ ISG, denoted as $S_0$. 

To find the unmasked and masked stabilizers from $\langle S_0\rangle$, we construct two sets $C$ and $V$. $C$ will keep track of the evolution of the stabilizer generators from $S_0$ and $V$ will keep track of the measurements from the measurement sequence. Both sets will be updated by the measurements in the sequence. $C_i$ refers to the set after the $i^{\mathrm{th}}$ round of measurements. We start with $C_0 = S_0$, the initial set of stabilizer generators that are linearly independent at round 0. $V$ is initialized as the empty set at round 0. $V_i$ refers to the set $V$ at the time after the $i^{\mathrm{th}}$ round of measurements. The union of $C_i$ and $V_i$ gives an overcomplete set of generators for the ISG $S_i$ at round $i$ or after the $i^{\mathrm{th}}$ round of measurements (See Figure \ref{fig:floquet code}). We use $C$ and $V$ when we are referring to the sets at the current time-step if the subscript is not specified. For example, at the time step right before the $i^{\mathrm{th}}$ round of measurements $C=C_{i-1}$, $V = V_{i-1}$.  The details of how both sets are updated along with the outcomes of their elements will be provided in the algorithm. 

Elements in $C$ and $\tilde{U}$ have stabilizer generators from $\langle S_0\rangle$ associated with them. For example, if $s\in C$ has $s_0$ associated with it, then $s\cdot s_0$ is a product of measurements and has a known outcome denoted by $O(s\cdot s_0)$. 

Each element $m$ in $V$ has a known outcome denoted by $O(m)$. In general, elements in $V$ can be found in $\langle \cup_{i=1}^{l} M_i \rangle $ and are not necessarily the product of measurements from any particular round. 

It is also important to note that for simplicity, the elements in $C$ and $V$ are unsigned Pauli operators. The signs are taken care of in the outcomes of the measurements: If a measurement has a negative sign, the measurement is tracked with a positive sign with the sign of its outcome flipped. 

$\tilde{U}$ and $\tilde{P}$ are initialized as empty sets to record down the relevant elements from the algorithm. The key ingredient is that we can find all the unmasked stabilizers from $\langle S_0\rangle$ through finding the generators of $\langle C \rangle \cap \langle V \rangle $ which are added to $\tilde{U}$ and then used to compute the unmasked stabilizers which form the set $U$. During the updates from the measurements in the sequence, a stabilizer can become permanently masked and is added to the set $\tilde{P}$. This is then used to compute the set of permanently masked stabilizers $P$ and their destabilizers $K$. By the end of the sequence, $U$, $P$ and $K$ are the output sets which give a complete set of unmasked stabilizers with their syndromes,  permanently masked stabilizers and their destabilizers respectively. We also obtained the set of temporarily masked stabilizers $T$ from $C_l$. $\langle U \cup T \cup P\rangle = \langle S_0 \rangle $ and $U \cup T\cup P$ is a complete and linearly independent set of stabilizers for the $0^{\mathrm{th}}$ ISG.

With this notation in hand, we describe the distance algorithm, by which the classification of stabilizers can be performed in $ O(n^3) $ time. The Zassenhaus algorithm suggested below requires $ O(n^3) $ time, whereas the rest of the algorithm requires $ O(n^2) $ time. 

\subsection{The Distance Algorithm}
\label{subsec: algorithm}

\textbf{Input:} A measurement sequence $M_1, M_2, \cdots, M_l$ on the n qubit code with initial ISG $S_0$ which defines the dynamical code $\mathcal{Q}$.
\\
\textbf{Output:} The set of permanently masked stabilizer generators $P$ and the set of destabilizers $K$, the set of unmasked stabilizer generators and their syndromes $U$, the set of temporarily masked stabilizer generators $T$ and the unmasked distance $d_u$ of the $0^{\mathrm{th}}$ ISG of $\mathcal{Q}$.

\begin{enumerate}

\item Initialize $C = S_0$ and also initialize $\tilde{U}$, $\tilde{P}$ and $V$ as empty sets.
 
\item For all rounds of measurements $i = 1,2,\cdots, l$, do the following:  

\begin{enumerate}

     \item For all measurements $m$ in round $i$, we update the sets $C$ and $V$ according to one of the 4 cases, depending on the commutation relations of $m$ with the elements in $C$ and $V$.

\begin{adjustwidth}{1.25 cm}{}
    Case 1: $[c,m]=[v,m]=0, \forall c \in C, \forall v \in V$
        \begin{adjustwidth}{0.5cm}{}
        Update rule: Add $m$ to $V$ if $m\notin \langle V\rangle $ and record down $O(m)$.
        \end{adjustwidth}   
        
        Case 2: $\{c,m\} =[v,m]=0$, for some $c \in C$ and $\forall v  \in V$
        \begin{adjustwidth}{0.5cm}{}
        Update rule: Add $m$ to $V$ and record $O(m)$. Remove an element $c_j$ in $C$ where $c_j$ anti-commutes with $m$. Then, for all $i$ such that elements $c_i\in C$ anti-commute with $m$, do the following update $c_i \to c_i \cdot c_j$.

        For each update of the form $c_i \to c_i \cdot c_j$: If $c_i $ has the stabilizer generator $s_i \in \langle S_0 \rangle$ associated with it and $c_j$ has the generator $s_j \in  \langle S_0 \rangle$ associated with it, then $c_i \cdot c_j$ has the generator $s_i \cdot s_j$ associated with it. Further, the outcome of $c_i \cdot c_j \cdot s_i \cdot s_j$ is recorded down as a product of $O(c_i \cdot s_i)$ and $O(c_j \cdot s_j)$. 
        
        For the element $c_j$ which is removed, denoting its associated stabilizer by $s_j$, add $c_j$ to $\tilde{P}$.
        
        \end{adjustwidth} 
        
        Case 3: $[c,m]=\{v,m\}=0 , \forall c \in C$ and for some $v \in V$
        
        \begin{adjustwidth}{0.5cm}{}
        
        Update rule: Update $V$ according to the stabilizer update rules and record the outcomes for the updated elements in $V$. Explicitly, add $m$ to $V$ and record $O(m)$ for $m$. Remove an element $m_j$ in $V$ where $m_j$ anti-commutes with $m$. Then, for all $i$ such that the elements $m_i \in V$ anti-commute with $m$, do the update: $m_i \to m_i \cdot m_j$. 

        For each updated element of the form $m_i \to m_i \cdot m_j$: $m_i \cdot m_j$ has the outcome $ O(m_i) \cdot O(m_j)$.
        
        \end{adjustwidth}
        
        Case 4: $\{c,m\}=\{v,m\}=0$ for some $c\in C$ and for some $v \in V$
        
        \begin{adjustwidth}{0.5cm}{}
        	
        Update rule: Add $m$ to $V$ and record down $O(m)$ for $m$. Remove an element $m_j$ in $V$ where $m_j$ anti-commutes with $m$. Then, for all elements $m_i\in V$ that anti-commute with $m$, do the following update: $m_i \to m_i \cdot m_j$ with $O(m_i)\cdot O(m_j)$ as its outcome. For all elements $c_i \in C$ that anti-commute with $m$, do the following update: $c_i \to c_i \cdot m_j$. If $s_i\in \langle S_0 \rangle$ is the stabilizer associated with $c_i$, then the stabilizer associated with $c_i \cdot m_j$ is $s_i$. If the operator and its outcome for $c_i \cdot s_i$ is given by $\tilde{m},O(\tilde{m})$ respectively, then for $c_i \cdot m_j$, the operator $\tilde{m}\cdot m_j$ has an outcome given by $O(\tilde{m})\cdot O(m_j)$.
        
        \end{adjustwidth}
    
\end{adjustwidth}  

    \end{enumerate} 
\item Calculate the generators for $\langle C \rangle \cap \langle V \rangle = \tilde{U}$ (A suggested algorithm for this is the \hyperlink{Zassenhaus algorithm}{Zassenhaus algorithm} in Subsection \ref{sub: unmasked and zassenhaus alg}.). If an element in $C$ is updated to an identity, it will be added to $\tilde{U}$ too. If $ C $ has redundant generators, then each redundant generator gives an identity that will be added to $\tilde{U}$: Suppose $ c_1 \in C $ is redundant and $ c_1 = \prod_{\alpha} c_\alpha$. Then, $ c_1 \cdot \prod_{\alpha} c_\alpha  = I$ is recorded into $\tilde{U}$.

\item Compute the unmasked stabilizers with their syndromes for the elements of $\tilde{U}$ and add them to $U$, using algorithm \ref{alg: unmasked stabilizers}. 

\item Compute the permanently masked stabilizers and their destabilizers in the $0^{\mathrm{th}}$ ISG using algorithm \ref{alg: pms and destabilizer}. The permanently masked stabilizers in the $0^{\mathrm{th}}$ ISG form the ordered set $P$, and their corresponding destabilizers form the ordered set $K$, such that the $j^{\mathrm{th}}$ element in $K$ is the destabilizer for the $j^{\mathrm{th}}$ element in $P$.

\item Compute the set of temporarily masked stabilizers $T$ using algorithm \ref{alg: tms}.

\item Compute the unmasked distance $d_u$ of the $0^{\mathrm{th}}$ ISG of $\mathcal{Q}$ given by Definition \ref{def:unmasked distance}. 
\end{enumerate}

\begin{remark}
    Elements in $C$ can look the same but may not be identical, as they may have different stabilizer generators associated with them. Thus, they cannot be treated as the same elements, but must be separately updated. For example, given an initial set of stabilizer generators consisting of $s_0^{(i)}$ and $s_0^{(j)}$, where the subscripts indicate that they are generators in the $0^{\mathrm{th}}$ round, we can perform some rounds of measurements. Following the steps in the algorithm, we may end up with $s_k^{(i)} = s_k^{(j)}$ after the $k^{\mathrm{th}}$ round of measurements, with $s_0^{(i)}$ updated to $s_k^{(i)}$ and $s_0^{(j)}$ updated to $s_k^{(j)}$. $s_k^{(i)}$ has an associated stabilizer $s_0^{(i)}$ but $s_k^{(j)}$ has an associated stabilizer $s_0^{(j)}$, and the known measurement outcomes given by $O(s_0^{(i)}\cdot s_k^{(i)})$ for $s_k^{(i)}$ and $O(s_0^{(j)}\cdot s_k^{(j)})$ for $s_k^{(j)}$ are clearly different. If we measure $p = s_k^{(j)} = s_k^{(i)}$, then we can separately obtain the syndromes for both $s_0^{(i)}$ and $s_0^{(j)}$, given by $O(p)\cdot O(s_0^{(i)}\cdot s_k^{(i)})$ and $O(p) \cdot O(s_0^{(j)}\cdot s_k^{(j)})$ respectively. In this case, despite looking alike, the original stabilizers associated with them are different, so they are tracked in $C$ as different stabilizer generators with separate labels. Only for Case 2 do we see the number of elements in $C$ decreases by 1. 
\end{remark}

\begin{remark}
    In our update of the measurements and their outcomes for the elements of $V$ and $C$, we can choose to not compute the products, but keep a running list of the measurements and their outcomes. This can be useful if we want to record information of the constituent measurements that give the syndrome information of the stabilizers in $\langle S_0 \rangle$.
\end{remark}

\begin{remark}[Temporarily masked stabilizers]
A temporarily masked stabilizer in $T$ can potentially be unmasked or permanently masked by increasing the measurement window.

\end{remark}

\begin{remark}[Comparison with outcome code \cite{delfosse2023spacetime}]
    The syndrome information obtained through the distance algorithm corresponds to the outcome code but limited to one round of ISG. Applying the distance algorithm to all ISGs gives the syndromes for all the unmasked stabilizers and this is exactly the outcome code. Here, we have provided a concrete algorithm to obtain the syndromes of the unmasked stabilizers for each ISG. 
    
    However, the main difference is that the distance algorithm distinguishes between the temporarily masked and the permanently masked stabilizers in an ISG. This is not present in the outcome code. Based on the type of masking, we can obtain the destabilizers for the masked stabilizers of each ISG, which allow us to determine a tighter upper bound in the distance of the dynamical code given by the unmasked distance.

\end{remark}
\subsubsection{Unmasked Stabilizers and their Syndromes} \label{sub: unmasked and zassenhaus alg}

\begin{algorithm}\label{alg: Zassenhaus algorithm} \hypertarget{Zassenhaus algorithm}{Zassenhaus algorithm}.
\\
\textbf{Input:} $C,V$.
\\
\textbf{Output:} Generators of $\langle C \rangle \cap \langle V \rangle$

Let  $ \tilde{V} $ be a basis set for $ \langle V \rangle $.

\begin{eqnarray*}
	&&\tilde{V} =\{v_1, \cdots, v_m\}\\
	&&C = \{c_1,\cdots, c_k\}
\end{eqnarray*}

Using the map $F$ defined in Eq. \eqref{eq: pauli to binary}, the elements $c_i, v_i$ are expressed as 2n dimensional vectors $F(v_i) =[a_i, b_i],F(c_i) =[d_i, e_i] $.

The following block matrix of size $((m+k)\times 4n)$ is created: 
\begin{equation}
	\begin{pmatrix}
		a_{1}&b_{1}&a_{1}&b_{1}\\
		\vdots&\vdots&\vdots&\vdots\\
		a_{m}&b_{m}&a_{m}&b_{m}\\
		d_{1}&e_{1}&0&0\\
		\vdots&\vdots&\vdots&\vdots\\
		d_{k}&e_{k}&0&0
	\end{pmatrix}
\end{equation}

Using elementary row operations, this matrix is transformed to the row echelon form:

\begin{equation}\label{eq: row echelon form}
	\begin{pmatrix}
		\alpha_{1}&\beta_{1}&\bullet&\bullet\\
		\vdots&\vdots&\vdots&\vdots\\
		\alpha_{q}&\beta_{q}&\bullet&\bullet\\
		0&0&e_1&f_1\\
		\vdots&\vdots&\vdots&\vdots\\
		0&0&e_j&f_j\\
		0&0&0&0\\
		\vdots&\vdots&\vdots&\vdots\\
		0&0&0&0
	\end{pmatrix}
\end{equation}

Here, $[\alpha_{i},\beta_{i}]$ is a non-zero 2n-dimensional vector and $\bullet$ stands for an arbitrary n-dimensional vector. 

The generators of $\langle C \rangle \cap \langle V \rangle$ are given by the set $\{ F^{-1}([e_i,f_i])\}, 1\leq i\leq j$ and the identities from the $ 0 $ vector rows of Eq. \eqref{eq: row echelon form}, as each $ 0 $ vector corresponds to a different associated stabilizer generator. 

\end{algorithm}

\begin{algorithm} \label{alg: unmasked stabilizers} Unmasked stabilizers.
\\
\textbf{Input:} $\tilde{U}$.
\\
\textbf{Output:} $U$

\begin{enumerate}

\item For all elements $u_j$, $j=1,2,\cdots$, in $\tilde{U}$, do the following:
    \begin{enumerate}
    	\item Let $ s_j $ be the stabilizer generator in $\langle S_0 \rangle$ associated with $u_j$. Add $s_j$ to $U$ if $s_j \notin \langle U\rangle$. 

        \item Let the outcome of $u_j$ be $O(u_j)$. Since $ u_j \in \langle V\rangle$, $O(u_j)$ can be computed using products of outcomes of measurements in $V$. The outcome for $u_j \cdot s_j$ is also known since this is a product of measurements with known outcomes. The syndrome information for $s_j$ is given by:

        \begin{eqnarray}
        	O(s_j) &=& O(u_j)\cdot O(u_j\cdot s_j) 
        \end{eqnarray}
        Associate this syndrome information with $ s_j \in U$
    \end{enumerate}
\item Return the updated set $U$.
\end{enumerate}
\end{algorithm}

\subsubsection{Permanently Masked Stabilizers and their Destabilizers}

We show how to find the destabilizers corresponding to the elements in $\tilde{P}$ in the $0^{\mathrm{th}}$ ISG.

Recall that to update a logical operator from the $i^{\mathrm{th}}$ ISG to the $(i+1)^{\mathrm{th}}$ ISG, we multiple a representative of the logical operator with elements in $\langle S_i \rangle$ so that the new representative commutes with both $\langle S_i \rangle$ and $\langle S_{i+1} \rangle$. This follows from the logical update rules in Lemma \ref{lemma: logical update rule}. To reverse this process, i.e. obtain the logical operator of the $i^{\mathrm{th}}$ ISG from a representative of the logical operator in the $(i+1)^{\mathrm{th}}$ ISG, we simply multiply to it elements in $\langle S_{i+1} \rangle$ so that the new representative commutes with both ISGs. 

Given an element $p\in \tilde{P}$ and its destabilizer $\kappa$ that we found through one of the rounds of measurements, we can treat this pair as logical representatives for a gauge qubit. Then, we can apply the reverse logical update rules to obtain the logical representatives for this gauge qubit in $S_0$. The details for this are provided as an algorithm below. 

\begin{algorithm}\label{alg: pms and destabilizer} Permanently masked stabilizer and its destabilizer.
\\
\textbf{Input:} The measurement sequence $M_1,M_2,\cdots, M_l$ on the n qubit code $ \mathcal{Q} $ with initial ISG $S_0$ and the set $\tilde{P}$.
\\
\textbf{Output:} $ P,K $
\begin{enumerate}
	
\item  Define a new set $R_0 = S_l$, where $S_l$ is the set of stabilizer generators for the $l^{\mathrm{th}}$ ISG.  

\item For all rounds of measurements $j = 1,2,\cdots, l$, do the following:  

\begin{enumerate}
     \item Initialize $r_{l-j+1}$ as an empty set.
     \item For all elements that are removed in step 2 of the distance algorithm using the measurements in $M_j$, add them to $r_{l-j+1}$. 
\end{enumerate}

\item Define the dynamical code $\mathcal{Q}_2$ on the n qubit code initialized as $R_0$ with the measurement sequence given by $r_1, r_2, \cdots, r_l$. $R_i$ denotes the $i^{\mathrm{th}}$ ISG of $\mathcal{Q}_2$. Let $R$ denote the set of stabilizer generators at the current time step. 

\item Initialize $L$ as an empty set to keep track of the gauge operators for $\mathcal{Q}_2$ .

\item For all rounds of measurements $r_j$ with $j = 1, \cdots, l$, do the following: 

	\begin{enumerate}
		\item For all measurements $m_i $ in $r_j$ with $i = 1, \cdots, q$:
		
			\begin{enumerate}

				\item If $m_i $ is an element in  $\tilde{P}$, that is if $m_i$ is removed from $\langle C \rangle$ using the measurements in $M_{l-j+1}$, then update $L$ and $R$ using the logical update rules in Lemma $\ref{lemma: logical update rule}$ and stabilizer update rules in Lemma \ref{lemma: stab update rule} respectively. Explicitly, remove an element $s$ in $R$ where $s$ anti-commutes with $m_i$. For all elements $ p $ in $ L, R$ that anti-commute with $m_i$, update $ p \to p \cdot s $. $ (m_i,s) $ forms the logical operators for a gauge qubit and is added to $L$. 
				
				\item If $m_i$ is not an element in $\tilde{P}$, then update the logical operators in $L$ and stabilizer generators in $R$ according to the logical update rules in Lemma $\ref{lemma: logical update rule}$ and the stabilizer update rules in Lemma \ref{lemma: stab update rule} respectively. 
			\end{enumerate}  
	\end{enumerate}

\item The permanently masked stabilizers and their destabilizers in the $0^{\mathrm{th}}$ ISG are given by the pairs of operators in $L$, where the first element in each pair is added to $P$ and the second added to $K$.

\end{enumerate}
\end{algorithm}

\subsubsection{Temporarily Masked Stabilizers}
\begin{algorithm} \label{alg: tms} Temporarily masked stabilizers.
\\
\textbf{Input:} $C_l, U$.
\\
\textbf{Output:} The set of temporarily masked stabilizer generators $T$.

\begin{enumerate}

\item Initialize $T$ as an empty set. 

\item For all elements $c_j$, in $C_l$, do the following:
    \begin{enumerate}
        \item Let $s_j$ be the stabilizer generator in $\langle S_0 \rangle$ associated with $c_j$. If $s_j \notin \langle U\rangle$, add $s_j$ to $T$.
    \end{enumerate}
\item Return $T$.
\end{enumerate}
\end{algorithm}

\subsection{Examples for the Distance Algorithm}
\label{subsec:examples}
\begin{example}[Illustration using Example \ref{ex: anti-commute}]
	Consider the stabilizer $ s = x_1 x_2 x_3 x_4 x_5 x_6 $ and the measurement sequence given by $ M_1 =  \{x_5x_6\}$, $ M_2 = \{z_6z_7\} $ and $ M_3 =\{x_1x_2, x_3x_4\} $. Then, applying the algorithm, we find the following: 
	
	\begin{center}
		\begin{tabular}{ |p{1.5cm}|p{3cm}|p{5cm}|  }
			\hline
			Round & $C$ & $V$  \\
			\hline 
			$0$ & $x_1 x_2 x_3 x_4 x_5 x_6$ &  \\
			\hline
			$1$ & $x_1 x_2 x_3 x_4 x_5 x_6$ & $x_5x_6$\\
			\hline
			$2$ & $x_1 x_2 x_3 x_4 $ & $z_6z_7$ \\
			\hline
			$3$ & $x_1 x_2 x_3 x_4 $ & $z_6z_7,x_1x_2, x_3x_4$ \\
			\hline
		\end{tabular}
		
	\end{center}
	$\langle C\rangle \cap \langle V\rangle = \{x_1 x_2 x_3 x_4 \}$

	Using algorithm \ref{alg: unmasked stabilizers}, we can substitute into the equation $ O(s_j) = O(u_j)\cdot O(u_j\cdot s_j) $ with $ O(u_j) = O(x_1 x_2 x_3 x_4 ) $ and $ O(u_j\cdot s_j) = O(x_5x_6) $. This gives the outcome for $ x_1 x_2 x_3 x_4 x_5 x_6 $ which is an unmasked stabilizer for this measurement sequence. 
\end{example}
\begin{example}[Illustration using a simplified honeycomb code]

    \begin{figure}[ht]
    \centering
    \includegraphics[width=0.35 \linewidth]{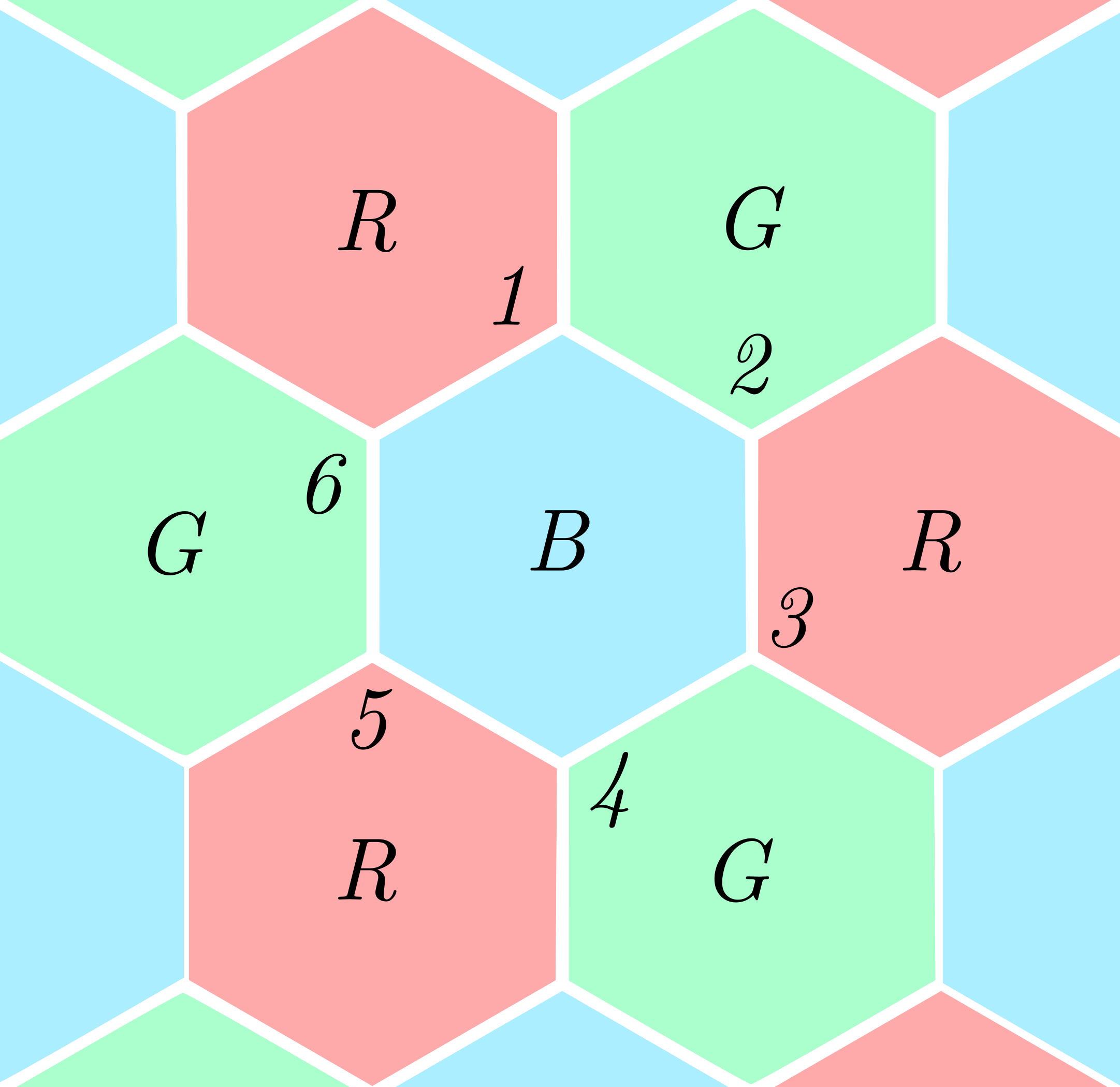} 
    
    \caption{An illustration of the three colorable honeycomb lattice. The plaquettes are labelled R for red, B for blue and G for green. The vertices of the center blue plaquette are labelled from 1 to 6.}
    \label{fig:illustration dist alg} 
    \end{figure}

    Consider a code built on a honeycomb lattice illustrated in Figure \ref{fig:illustration dist alg}. For the sole purpose of illustrating how the distance algorithm works, we focus on the center blue plaquette. Suppose the stabilizer $A = z_1 z_2 z_3 z_4 z_5 z_6 \in S_0$, and the dynamical code consists of measurement sequences $M_1 = \{x_1x_2,x_3x_4,x_5x_6\}$, i.e. the lines connecting the red plaquettes, and $M_2 = \{y_2y_3,y_4y_5,y_6y_1\}$, i.e. the lines connecting the green plaquettes. 

    Then, using the algorithm, we obtain the following table. 

    \begin{center}
    \begin{tabular}{ |p{1.5cm}|p{3cm}|p{5cm}|  }
    \hline
        Round & $C$ & $V$  \\
    \hline 
        $0$ & $z_1 z_2 z_3 z_4 z_5 z_6$ &  \\
    \hline
        $1$ & $z_1 z_2 z_3 z_4 z_5 z_6$ & $x_1x_2,x_3x_4,x_5x_6$\\
    \hline
        $2$ & $z_1 z_2 z_3 z_4 z_5 z_6$ & $x_1x_2x_3x_4x_5x_6,y_2y_3,y_4y_5,y_6y_1$ \\
    \hline
    \end{tabular}
    \end{center}
    Calculating $\langle C\rangle \cap \langle V\rangle$, we found the outcome for the blue plaquette in round $2$ of the code, with $O(A)=O(x_1x_2x_3x_4x_5x_6)O(y_2y_3)O(y_4y_5)O(y_6y_1)$. We conclude that the blue plaquette is an unmasked stabilizer in the $0^{\mathrm{th}}$ ISG. 
\end{example}
\begin{example}[Hastings and Haah's Honeycomb code \cite{Hastings2021dynamically}]
    Generalizing to this code from the previous example, we can substitute $ S_0 $ of the algorithm with that of any ISG post-initialization. The algorithm shows that the plaquettes of this ISG are the unmasked stabilizers, whose outcomes are fully revealed within a window of four rounds. The weight 2 checks in this ISG are the permanently masked stabilizers, with destabilizers given by the anti-commuting checks of the next round of measurements. To calculate the unmasked distance, the gauge group is defined with these anti-commuting checks as destabilizers.
\end{example}

\begin{example}[Dynamic automorphism codes \cite{davydova2023quantum}]
    Our algorithm can verify if a given dynamic automorphism code defined by a sequence of automorphisms of color code has good error correcting properties. This can be done by running the distance algorithm for each ISG of the code. We can specify a time window for obtaining the syndrome information and terminate the algorithm after a predefined number of rounds of measurements. The unmasked distance is then determined based on the set of unmasked, temporarily and permanently masked stabilizers. 
    
    In \cite{davydova2023quantum}, the authors introduced padding sequences between two automorphisms. One can observe that adding padding changes some of the temporarily masked stabilizers to unmasked stabilizers. For example, in the padded $\phi_{(\mathrm{gb})}$ automorphism, padding converts some temporarily masked stabilizers like $R_2^{(1)},R_2^{(2)}$ into unmasked stabilizers whose syndromes can be obtained through measurements in the padded rounds. 
\end{example}

\subsection{Proof for the Distance Algorithm}
\label{subsec: proof}
In this subsection, we prove several components of the distance algorithm. We show that it outputs a complete set of masked and unmasked stabilizer generators. We also show several results that are important in the construction of the algorithm and its proof. 

We first show that every stabilizer in $U$ gives a syndrome for the stabilizers in $\langle S_0\rangle $. Then, we show that for a given measurement protocol, the set of syndromes from $S_0$ that can be obtained using the distance algorithm is complete. 

\begin{theorem}\label{theorem: known outcome}
    Every element in $U$ is an unmasked stabilizer in $\langle S_0\rangle $ with known syndrome. 
\end{theorem} 

\begin{proof}
	Let $ u $ be an element in $ \tilde{U} $ and let its associated stabilizer be $ s_u $. We prove that both $ O(u) $ and $ O(u\cdot s_u) $ are known. This shows that $O( s_u )$  has a known syndrome given by $ O(u)\cdot O(u\cdot s_u) $.
	
	Since $u \in \langle V\rangle$, $O(u)$ is a known outcome obtained from the product of outcomes of measurements in $ V $. As $ u \in \langle C\rangle $, $ u $ has an associated stabilizer $ s_u \in \langle S_0 \rangle$.  If $ u = s_u $, then $ u \cdot s_u = I $ has known outcome $ +1 $.  
	
	To show that the outcome of $ u\cdot s_u $ is known, we first show that for any element $ c \in C $ with its associated stabilizer $ s_c $, $ c \cdot s_c$ always has a known outcome denoted by $ O(c \cdot s_c) $. We prove this by induction. 
	
	Suppose the statement is true for all elements in $ C $ before the $ k^{\mathrm{th}} $ measurements. 
	
	Then if we measure $ m $ as the $ (k+1)^{\mathrm{th}}$ measurement, an element $ c \in C $ can be updated using one of the four cases in Step 2 of the distance algorithm. In all 4 cases, the statement remains true.
	
	Case 1 and Case 3 do not update $ C$, so the statement is trivially true in these 2 cases. 
	
	In Case 2, if $ c_1 $ has an associated stabilizer $ s_{c_1}$ and $ O(c_1\cdot s_{c_1}) $ is known and $ c_2 $ has an associated stabilizer $ s_{c_2} $ and $ O(c_2 \cdot s_{c_2}) $ is known by the induction hypothesis. Then, if applying the update rule removes $ c_2 $ and updates  $  c_1 \to c_1 \cdot c_2 $, the associated stabilizer is updated to $ s_1 \cdot s_2 $, and the outcome $ O((c_1\cdot c_2 )\cdot (s_{c_1} \cdot s_{c_2})) $ is given by $ O(c_1\cdot s_{c_1})\cdot O(c_2 \cdot s_{c_2}) $ which is known. 
	
	In Case 4, if an element $ c $ in $ C $, with an associated stabilizer $ s_c $, is updated by the measurement $ m \in V $ such that $  c \to c \cdot m $, then the associated stabilizer remains the same as $ s_c $. The outcome $ O((c\cdot m) \cdot s_c) $ is given by $ O(c\cdot s_c) \cdot O(m) $ which is known. 
	
	Thus, we showed that for any element $ c \in C $ with an associated stabilizer $ s_c \in \langle S_0 \rangle $, $ O(c\cdot s_c) $ is always known. 
	
	Since $ O(s_u) = O(u\cdot s_u) \cdot O(u)$, and $ O(u\cdot s_u)  $ is a known outcome by the argument above, any element in $ \langle C \rangle \cap \langle V \rangle $ gives an unmasked stabilizer generator $ s_u\in \langle S_0 \rangle$ along with its syndrome $ O(s_u) $.
	
\end{proof}

The following lemma is needed to show the next lemma.
\begin{lemma}\label{lemma: basis indept}
	The group $ \langle C\rangle  \cap \langle V \rangle $ is independent of the choice of stabilizer generators for $ \langle C\rangle $ and $ \langle V\rangle $ up to some elements in $ \langle V \rangle $.
\end{lemma}
\begin{proof}
	We prove this by considering two choices of generators for $ C $ and $ V $, with the elements in the two choices for $ C $ differing by some elements in $ \langle V\rangle $. This is necessary since Case 4 of the distance algorithm will remove an arbitrary anti-commuting element from $ V $ and multiply it to some elements in $ C $. We show that for both choices, we obtain the same set of unmasked stabilizers.  One property that we will make use of heavily here is that the stabilizer update rules in Lemma \ref{lemma: stab update rule} are basis independent. 
	
	Let $ V_1,V_2 $ be two arbitrary bases for $ \langle V \rangle$ such that $ \langle V_1\rangle =  \langle V_2\rangle = \langle V\rangle$. Let $ C_1,C_2 $ be two sets for $ C $. Elements in $ C_1 $ can be expressed in terms of elements in $ \langle C_2 \rangle $ and some generators $ \{v_1, \cdots, v_k\} \subset \langle V \rangle $ and vice versa. Note that for $ C_1 $ and $ C_2 $, we must be careful that there may be elements that have the same Pauli operator but different associated stabilizers, thus we require that in our choices the groups formed by the associated stabilizers of the elements in the two sets must be equal, that is for any element $c_1 \in \langle C_1\rangle$, $\exists c_2 \in \langle C_2\rangle$ with the same associated stabilizer generator as $c_1$ such that $c_1 = c_2 \cdot v$, with $v \in \langle V\rangle$, and vice versa.
	
	Initially, both choices of generators are equivalent in finding the unmasked stabilizers. Next, we show that this holds true after an arbitrary measurement $ m $, by checking for all 4 cases of the commutation relations on $ C_1,V_1 $.  We first check for Case 1 and Case 2. Then, we check the update when it is Case 3 for both $ C_1,V_1 $ and $ C_2,V_2 $. Lastly, we check the update when it is Case 3 for $ C_1,V_1 $ and either Case 3 or 4 for $ C_2,V_2 $ without loss of generality. 

	If the update is according to Case 1 for $ C_1, V_1 $, it must also be updated by Case 1 for $ C_2, V_2 $. $ m $ is added to $ V_i $, so $ \langle V_1\rangle = \langle V_2 \rangle $. It is clear that elements in $ \langle C_1\rangle  \cap \langle V_1 \rangle $ still differ from $ \langle C_2 \rangle  \cap \langle V_2 \rangle $ up to $ \langle \{v_1, \cdots, v_k\} \rangle $ and are equivalent in finding the unmasked stabilizers. 
	
	If it is Case 2 for $ C_1,V_1 $, then it must be Case 2 for $ C_2, V_2 $. We add $\{ v_1,\cdots, v_k \}$ to both $ C_1$ and $C_2$, so that they generate the same group, and update them using the stabilizer update rules in Lemma \ref{lemma: stab update rule}, then the sets are updated to $ \{v_1,\cdots, v_k\} \cup C_i' \cup \{m\}$, with $ C_i \to C_i' $, $i \in \{1,2\}$. The updated sets still generate the same group as the stabilizer update rules are basis independent. Removing $ \{v_1,\cdots, v_k\}\cup  \{m\}$ from both sets, it is clear that elements in $ C_i'$ are equivalent up to elements in $\langle \{v_1,\cdots, v_k\} \rangle$. $ C_i' $ are exactly the updated $ C_i $ where $ \{v_1,\cdots, v_k\}$ has not been added to the sets. Since $  \langle V_1\rangle = \langle V_2 \rangle $ and $ \{v_1,\cdots, v_k\} \in \langle V_i \rangle $ after the update, an element in $\langle C_1\rangle $ still differs from $ \langle C_2 \rangle$ by elements in $ \langle \{v_1, \cdots, v_k \}\rangle$.
 
    Additionally, the associated stabilizer generators for the updated $C_i$ generate the same group. If $s$ is an associated stabilizer generator such that it is associated with $c_1\in \langle C_1 \rangle$ and $c_2 \in \langle C_2 \rangle$ and $c_1= c_2 \cdot v$ for some $v \in \langle V\rangle$ before the update, then $[v,m]=0$ implies that either both $c_1$ and $c_2$ are in the updated groups $\langle C_1'\rangle$ and $\langle C_2'\rangle$ respectively or they both anti-commute with $m$ and are not in the respective updated groups, in which case $s$ is not in the group generated by the associated stabilizers for both $C_i'$.  
    
    Therefore, after the update, for an element $c_1$ in $ \langle C_1\rangle  \cap \langle V_1 \rangle $ that has an associated stabilizer generator $s$, we can still find $c_2 \in \langle C_2 \rangle$ with the same associated stabilizer $s$, and since $c_1$ and $c_2$ differ by an element $v \in \langle V\rangle$ before the update, we still have $v\in V$ after the update, so $c_2 \in \langle C_2 \rangle  \cap \langle V_2 \rangle $ after the update.
	
	If it is Case 3 for both $ C_1,V_1 $ and $ C_2,V_2 $, the updates only change $ V $.  Since the update for $ V $ is the same as the stabilizer update rules with $ V $ as the set of stabilizers, this update is independent of the choice of basis for $ V $, so after the update, $ \langle V_1\rangle = \langle V_2 \rangle $. For any element $ c_1 $ in $  \langle C_1\rangle $ with associated stabilizer generator $s_c$, $ c_1 \cdot v_{c_1} \in \langle C_2\rangle$ for some $ v_{c_1}  \in \langle v_1, \cdots, v_k  \rangle $ with the same associated stabilizer generator. Then this implies $ [m,v_{c_1}] =0 $, otherwise, $ m $ anti-commutes with some elements in either $ C_1 $ or $ C_2 $. Therefore, $\{v_1, \cdots, v_k \} $ remains in $ \langle V_i \rangle$, so it follows that the statement remains true for Case 3. 
	
	If it is Case 4 for $ C_1,V_1 $ and either Case 3 or 4 for $ C_2, V_2 $, by the same reasoning as in Case 3, after the update by $ m $, $ \langle V_1\rangle = \langle V_2 \rangle $. Suppose an element $ a \in V_1 $ is removed and an element $ b\in V_2 $ is removed. By the stabilizer update rules, some elements in $ C_1 $ are updated with $ a $ and some in $ C_2 $ are updated with $b $, so the group generated by the associated stabilizer generators remains the same for both $C_i$. The update works the same for $ C_1, C_2 $ if we consider updating $G_1= \{a, b, v_1, \cdots, v_k\} \cup C_1 $ and $G_2 = \{b, a, v_1, \cdots, v_k\} \cup C_2  $ as the stabilizer sets. Updating by the measurement $ m $ such that $ a $ is removed from $ G_1 $ and $ b $ is removed from $ G_2 $, we have $ G_1 \to G_1' = \{m, ab, v_1', v_2',\cdots, v_k'\}\cup C_1' $ and $  G_2 \to G_2' = \{m, ab, \tilde{v}_1, \tilde{v}_2,\cdots,\tilde{v}_k\}\cup C_2' $, where $ v_i \in G_1$ is updated to $ v_i' $ and $ v_i \in G_2 $ is updated to $ \tilde{v}_i $. 
	
	First note that $ G_1, G_2 $ generate the same group, so after the update, $ \langle G_1'\rangle =\langle G_2'\rangle $. Further, the updated $ C_i $ is the same as $C_i'$, as both are updated by the same elements $ a $ or $ b $. Since $\{b, a, v_1, \cdots, v_k\} \subset \langle V\rangle $, after the update, both $\{m, ab, v_1', v_2',\cdots, v_k'\}$ and $  \{m, ab, \tilde{v}_1, \tilde{v}_2,\cdots,\tilde{v}_k\} $ are subsets of the updated $ \langle V\rangle $. 
	
	If we remove $\{m, ab, v_1', v_2',\cdots, v_k'\}$ and $  \{m, ab, \tilde{v}_1, \tilde{v}_2,\cdots,\tilde{v}_k\} $ from $ G_1'$ and $ G_2'$ respectively, then the elements in $ C_1'$ differ from those in $ C_2'$ by some elements in $ \{ab\}\cup \{v_1', v_2',\cdots, v_k'\}\cup \{\tilde{v}_1, \tilde{v}_2,\cdots,\tilde{v}_k\} $ which is a subset of the updated $\langle V\rangle$. Therefore, every element in the updated $C_1 $ can still be expressed as elements in the updated $ C_2 $ and elements in the updated $ \langle V\rangle $, and vice versa.  

    One can also show that if $s$ is an associated stabilizer generator for $c_1 \in \langle C_1\rangle$ and $c_2 \in \langle C_2\rangle$ with $c_1 = c_2 \cdot v$ before the update then $s$ is still in the associated stabilizer group after the update and $c_1' = c_2' \cdot v'$, with $v' $ in the updated $ \langle V\rangle$, where $s$ is associated with both $c_1'\in \langle C_1'\rangle,c_2' \in \langle C_2'\rangle$. The proof is as follows:

    Since the stabilizer update rules are basis independent, we have $c_1 \cdot a^{\alpha} \in \langle C_1' \rangle $ and $c_2 \cdot b^{\beta} \in \langle C_2' \rangle $, with $\alpha, \beta \in \{0,1\}$ depending on the commutation relations. We consider 2 cases: $[v,m] =0$ and $\{v,m\}=0$.

    In the first case, $c_1$ and $c_2$ either both commute or anti-commute with $m$, so $\alpha = \beta$. 
    \begin{eqnarray*}
        c_1 &\to& c_1' = c_1 \cdot a^{\alpha}\\
        c_2 &\to& c_2' = c_2 \cdot b^{\alpha}
    \end{eqnarray*}
        
    $c_1'\cdot c_2' = v\cdot (ab)^{\alpha} \in \langle V\rangle$, since $v \in \langle V\rangle$ and $ab \in \langle V\rangle $.

    In the second case, we instead obtain $\alpha = \beta +1 (\mathrm{mod}\: 2)$, but again $c_1'\cdot c_2' = v \cdot a^{\alpha}\cdot b^{\alpha+1} \in \langle V\rangle$, since both $v\cdot b$ or $v \cdot a$ is in $\langle V\rangle$ and $ab \in \langle V\rangle$. 
    
    In all 4 Cases, we show that $ \langle C_1\rangle  \langle V_1\rangle = \langle C_2\rangle \langle V_2\rangle = \langle C\rangle \langle V \rangle $ and $  \langle V_1 \rangle =  \langle V_2 \rangle =  \langle V\rangle $ after the update and we show that $\langle C_1\rangle  \cap \langle V_1\rangle$ gives the same set of unmasked stabilizers as $\langle C_2\rangle  \cap \langle V_2\rangle$. By induction, the set of unmasked stabilizers that one can find using $ \langle C\rangle  \cap \langle V\rangle $ is basis independent and we have the freedom for a change in the generators of $ C $ and $ V $ during any step of the distance algorithm. 
\end{proof}

\begin{lemma}\label{lemma: compute at the end}
    Suppose a stabilizer $ s_u \in \langle S_0\rangle $ is unmasked after some measurements through calculating the generators of $\langle C\rangle \cap \langle V \rangle$, then $ s_u $ can be unmasked through calculating $\langle C\rangle \cap \langle V \rangle$ after all future measurements. In particular, $ s_u $ can be unmasked at the end of the entire measurement sequence. 
\end{lemma}

\begin{proof}
    We want to show that finding $\langle C\rangle \cap \langle V \rangle$ can be done at the end of updating the sets $C$ and $V$ with all the measurements from that round. We show this by first assuming that a stabilizer is unmasked in the midst of the measurements, then we show that it will still be unmasked after the next measurement is made, regardless of what it is. 
        
    If the measurements consist of only those that fall under Case 1, then the statement holds trivially. 
	
	Let $ \{m_1, m_2,\cdots, m_q\} $ be a series of measurements. Suppose we were able to obtain the outcome for the stabilizer associated with $ u \in \langle C\rangle \cap \langle V \rangle$ after measuring $ m_i $ with $ 1\leq i\leq q $. Let $ s_u\in \langle S_0 \rangle $ be the unmasked stabilizer associated with $ u $. If the measurements end here, calculating $ \langle C\rangle \cap \langle V \rangle $ gives the syndrome for $s_u$. 
	
	Using Lemma \ref{lemma: basis indept}, the stabilizers for $ C$ and $ V $ are changed so that $ u $ is in the new sets of $ C $ and $ V $. 
   
    If we continue with the rest of the measurements in this series, then for a new measurement $m_{i+1}$, it either commutes or anti-commutes with $u$. In the first case, $u$ remains in $C$ and in $V$, so $u \in \langle C\rangle \cap \langle V \rangle$, and $ s_u $ is unmasked. In the second case, $\{m_{i+1}, u\}=0$, and $ C $ and $ V $ are updated according to Case 4 in Step 2 of the distance algorithm. Then, some element $ m \in V $ that anti-commutes with $ m_{i+1} $ is removed from $ V $. If $ m \neq u $, then $ u $ is updated to $ u \cdot m $ in both sets $ C $ and $ V $, so $ u\cdot m \in \langle C\rangle \cap \langle V \rangle $, giving the syndrome for $ s_u $. If $ m=u $, then the update gives $ m\cdot u = I $ in the set $ C $, so we can infer the outcome $ O(s_u) $ using the fact that $ O((m\cdot u)\cdot s_u = s_u) $ is a known outcome by the proof in Theorem \ref{theorem: known outcome}.  
    
    Therefore, we can calculate $\langle C\rangle \cap \langle V \rangle$ and obtain the syndrome for $ s_u $ after the $ (i+1)^{\mathrm{th}}$ measurement, if $ s_u $ can be unmasked using the first $ i $ measurements. We can conclude from here that we can simply calculate $\langle C\rangle \cap \langle V \rangle$ at the end of updating $ C $ and $ V $ with all the measurements to obtain the full set of unmasked stabilizers from the distance algorithm. 
\end{proof}


\begin{theorem}\label{theorem: sufficient}
	The distance algorithm outputs the entire set of unmasked stabilizers for $\langle S_0\rangle $ together with their syndromes. 
\end{theorem}
\begin{proof}
	By induction, suppose all the unmasked stabilizers that can be obtained in the first $k$ measurements can be found using the distance algorithm in Subsection \ref{subsec: algorithm}. We want to show that this also holds for the first $k+1$ measurements. 
	
	Suppose $s_0\in S_0$ is a stabilizer generator whose syndrome can be obtained using the first $k+1$ measurements, so $ O(s_0)$ is given by the outcome of some of the first $ k $ measurements and $ O(m_{k+1})$, where $O(m_i) \in {\pm 1}$ is the outcome of the measurement $m_i$. 
	
	Since $m_{k+1}$ must be measured last, the outcome of $ m_{k+1} $ must be deterministic. We know that $s_0$ is in the stabilizer set in round 0, and that its outcome is pre-determined, if no error occurs. This means that we know in the situation of no error on the code what the outcome of $s_0$ will be if we measure it at round 0. But because errors may have occurred, the rounds of measurements allow us to compare the two outcomes for $ s_0 $, to determine the errors that have occurred at round 0. Thus, the outcome of $m_{k+1}$ must be deterministic in nature if we have fixed the outcomes for $m_i$ with $i\leq k$ and errors $E$ on the code. 
	
	In other words, $m_{k+1}$ is in $\langle C\cup V \rangle$, prior to the $(k+1)^{\mathrm{th}}$ measurements. 
	
		\begin{eqnarray}
		&& m_{k+1} \in \langle S \rangle = \langle C \cup V\rangle\\
		\implies&& m_{k+1} = c_{k+1}\cdot v_{k+1} \\
		&&\mathrm{where}\; c_{k+1} \in \langle C \rangle , v_{k+1} \in \langle V \rangle
	\end{eqnarray}
	
	The outcome for $c_{k+1}$ can be obtained from measuring $ m_{k+1} $. $c_{k+1}$ must not be identity operator, otherwise $ O(s_0) $ can be obtained from the first $ k $ measurements. Let $ s_{k+1} $ be the stabilizer associated with $ c_{k+1} $, then $s_{k+1} \in \langle S_0 \rangle $ is an unmasked stabilizer generator. Since $ c_{k+1} =v_{k+1}\cdot m_{k+1} $, $ c_{k+1}\in \langle V\rangle $ too.
	
	 If $s_{k+1}  =  s_0$, then $ O(s_0) $ is found. For the case where $ s_{k+1} \neq s_0 $, we first note that the outcomes of $ s_{k+1} $ and $ s_0 $ are independent of each other, so to find $ s_{0} $ using $ s_{k+1} $, the only way is that $ s_0\cdot s_{k+1} $ has a known outcome using the first $ k $ measurements. 
	 
	 Since $ s_0\cdot s_{k+1} $ is unmasked in the first $ k $ measurements, by the induction hypothesis, $  s_0\cdot s_{k+1} $ can be found using the distance algorithm after some measurements. By Lemma \ref{lemma: compute at the end}, we know that if a stabilizer is unmasked by computing $ \langle C \rangle \cap \langle V\rangle  $ at some point, then it can be obtained by computing $ \langle C \rangle \cap \langle V\rangle  $ after all future measurements. Therefore, there exists some element $ \tilde{c} $ in $ \langle C \rangle$, such that it has the associated stabilizer  $ s_0\cdot s_{k+1} $, and $ \tilde{c} $ is in $ \langle C \rangle \cap \langle V\rangle$ after $ k+1 $ measurements.
	 
	 Then $ \tilde{c}\cdot c_{k+1} $ is an element in $ \langle C \rangle \cap \langle V\rangle$ with the stabilizer $ s_0 $ associated with it. Thus, we showed by induction that $ s_0 $ can be unmasked using the distance algorithm, if it is an unmasked stabilizer. 
	
\end{proof}

The following proposition shows that a stabilizer cannot be measured if it is removed from $C$ because of the update rules in Case 2 of the distance algorithm. 

\begin{proposition}
    \label{statement4} After a stabilizer $ c $ is removed from $C$, we permanently masked the syndrome of the stabilizer $s_c \in \langle S_0\rangle$ associated with $ c $ if it has not been unmasked from previous measurements.
\end{proposition}
\begin{proof}
	Suppose $ c\in C $ is removed after the $ k^{\mathrm{th}}$ measurement and $ s_c $ associated with $ c $ has not been unmasked by the distance algorithm. 
		
	Since Theorem \ref{theorem: known outcome} and Theorem \ref{theorem: sufficient} demonstrate that every unmasked stabilizer from the first $ k$ measurements can be obtained using the distance algorithm, $ s_c $ is either temporarily or permanently masked by the measurement sequence. 
	
	Applying Lemma \ref{lemma: compute at the end}, after the $ k^{\mathrm{th}} $ measurement, the set of stabilizers that can be unmasked by the distance algorithm are those associated with the generators of $ \langle C\rangle \cap \langle V \rangle$. If the measurement sequence ends here, then any generator that has not been unmasked but is associated with some stabilizer in $ \langle C\rangle  $ is temporarily masked. This is because we can measure and unmask it in the future, for example, by directly measuring the generator in the next round. Further, using the update rules in the distance algorithm, any stabilizer generator that is not associated with an element in $ \langle C\rangle $ at this point will still not be associated with any element after any additional sequence of measurements. Thus, it will not be unmasked by the distance algorithm. Applying Theorem \ref{theorem: sufficient} again, this implies that the generator cannot be unmasked by all future measurements and therefore is permanently masked.
	
	Since $ s_c $ is not associated with any generator in $ \langle C\rangle $ after the $ k^{\mathrm{th}}$ measurement, it is permanently masked. 
	
	From here, we can infer that the set of stabilizers from $ \langle S_0 \rangle $ that we can unmask from the given sequence and further measurements can only be the set $ \langle T \rangle \langle U \rangle$.  Therefore, the permanently masked stabilizers $ P $ is given by $ \langle S_0 \rangle \backslash \langle T\rangle \langle U\rangle $. This precisely corresponds to the stabilizers that are removed from $ C $ due to Case 2 in Step 2 of the distance algorithm with associated stabilizers that have not been unmasked yet. 

\end{proof}

\section{Initialization of a Floquet Code}\label{sec:init}

A special case of dynamical code is a Floquet code where the sequence of measurements is periodic. Each cycle in a Floquet code consists of measuring the same set of measurements. To understand the behavior of Floquet codes, it is important to learn how the stabilizer groups may evolve from cycle to cycle. Here, we show several interesting properties of Floquet codes including its code initialization and masking properties.

In this section, we do not consider the effects of errors on initialization.  We instead assume that the error-correcting properties of the Floquet code to initialize and maintain the code space.  Instead, we focus on the question of how the code space evolves under the sequence of measurements defining the code and how long it takes to reach a periodically repeating state.

The proofs for some of the results shown here become exceedingly simple by making use of the results and techniques developed for the distance algorithm. 

Let $ M $ be a measurement sequence, consisting of $ l $ measurements in the ordered set $ \{m_1,m_2,\cdots, m_l\} $. We define one cycle as performing the entire measurement sequence $ M $ once. Furthermore, we assume that it is an empty stabilizer group at the start of the first measurement cycle. Let the set of stabilizer generators after measuring $ m_i $ in the $ j^{\mathrm{th}} $ measurement cycle be denoted by $ S_j^{(i)} $. 

The following result holds for all cycles $ j $:
\begin{theorem} \label{theorem: subset_floquet}
The instantaneous stabilizer generators after measuring $ m_i $ in the previous cycle is a subset of the instantaneous stabilizer generators after measuring $ m_i $ in the current cycle, that is $ S_j^{(i)} \subseteq  S_{j+1}^{(i)} $. 
\end{theorem}

\begin{proof}
We prove this by induction. 
	
We first show that this is true for the first two cycles of measurements. It holds that $ S_1^{(1)}$ is a subset of $ S_2^{(1)} $, since $S_1^{(1)}$ consists of a single measurement $ m_1 $, and after measuring $ m_1 $ in the second cycle, $ m_1 $ must be in $ S_2^{(1)} $. To show that after each measurement, the updated $S_1^{(1)}$ will still be a subset of the updated $ S_2^{(1)} $, we apply step 2 of the Distance Algorithm with $ C$ = $S_2^{(1)}\backslash S_1^{(1)}$ and $ V= S_1^{(1)}$, with the measurement sequence given by $ M\backslash\{m_1\}$. After the $ i^{\mathrm{th}} $ measurement, $ V$ is updated to $ S_1^{(i)} $. Thus, $ S_1^{(i)} \subset V^{(i)} \cup C^{(i)}= S_2^{(i)}$, where $ V^{(i)} $ denotes the updated $ V $ and $ C^{(i)} $ denotes the updated $ C $, after the $ i^{\mathrm{th}} $ measurement.
	
Suppose this is true for the first $ k $ cycles, we want to show that it is true for the first $ k+1 $ cycles. By the induction hypothesis, $ S_{j-1}^{(i)} $ is a subset of $ S_j^{(i)} $, $\forall i \leq l$, $ \forall j<k-1 $. We want to show that $ S_k^{(i)} \subseteq S_{k+1}^{(i)},  \forall i $. Since $ S_{k-1}^{(l)} $ is a subset of $ S_{k}^{(l)} $, we can run step 2 of the distance algorithm to two sets of $ C,V $, initializing the first set as $ C_{a_0} = \{\}, V_{a_0} = S_{k-1}^{l} $, and the second set as $ C_{ b_0}  = S_{k}^{(l)}\backslash S_{k-1}^{(l)}, V_{b_0} = S_{k-1}^{(l)} $. We have $ V_{a_0}= V_{b_0}$, but $   C_{a_0} \subseteq C_{ b_0}  $ and $C_{b_0}\cup V_{b_0} = S_{k}^{(l)} $. Running the distance algorithm with the set of measurements $ M $, it is easy to see that $ V_{a_0} = V_{b_0} $ after all $ i $ measurements, and since $ C_{a_0} $ remains empty, $ V_{a_0} \cup C_{a_0} \subseteq V_{b_0} \cup C_{b_0} $. Therefore, $ S_{k}^{(i)} \subseteq S_{k+1}^{(i)} , \forall i $. 
	
\end{proof}

\begin{theorem}\label{theorem: new_floquet}
If there are $ k $ new generators from $  \langle S_{j+1}^{(0)} \rangle / \langle S_{j }^{(0)}\rangle$, then there can be at most $ k $ more new stabilizer generators added in the next cycle. Furthermore, the number of generators can only increase due to a measurement $ m_i $ in the $ (j+1)^{\mathrm{th}} $ cycle if in the $ j^{\mathrm{th}} $ cycle, the number of generators in the stabilizer group increased by 1 after measuring $ m_i $. 
\end{theorem}
	
\begin{proof}
	
The condition for the set of stabilizer generators to increase with a measurement $ m $ is $ [m, s] =0, \forall s \in  S$ and $  m \notin \langle S \rangle $.  Since $ S_{j }^{(i)}\subseteq S_{j+1}^{(i)}$, if the condition holds for $S_{j+1}^{(i)}  $, then it must also hold for $ S_{j }^{(i)} $, thus an increase in the number of stabilizer generators can only happen due to the set of measurements that causes the number of generators to increase in the previous cycle. 
\end{proof}

Theorem \ref{theorem: new_floquet} implies that in the worst case scenario, it takes $n-1$ cycles to initialize a Floquet code with $n$ stabilizer generators. 

The following theorem shows that the worst case scenario exists for some measurement sequence:

\begin{theorem}
There exists a measurement sequence for a Floquet code that takes $ n-1 $ cycles to fully initialize for a code with $ n $ stabilizer generators. 
\end{theorem}

\begin{proof}
We show this by giving a recursive construction of measurement sequence such that it takes $ n-1 $ cycles to fully initialize the code with $ n $ stabilizers. 

Let $ \{s_i: 1 \leq i\leq n\} $ denote a set of independent stabilizer generators for the fully initialized code. Let $ d_{s_i} $ denote the destabilizer for $ s_i $. The construction works in a way that $ s_i$ is added to the set of stabilizer generators after $ i-1 $ cycles and each cycle increases the number of generators by 1 even though the same set of measurements is made in each cycle. The rest of the generators are updated in a way that $ s_1 $ is mapped to $ s_1 $ and $ s_i $ is mapped to $ s_{i+1} $ after one cycle. The mapping is simply a way to keep track of how the stabilizer generators change with measurements as shown in Figure \ref{fig: init_floquet} for illustration. 
	
	\begin{figure}[ht]
		\centering
		\includegraphics[width= 0.75 \linewidth]{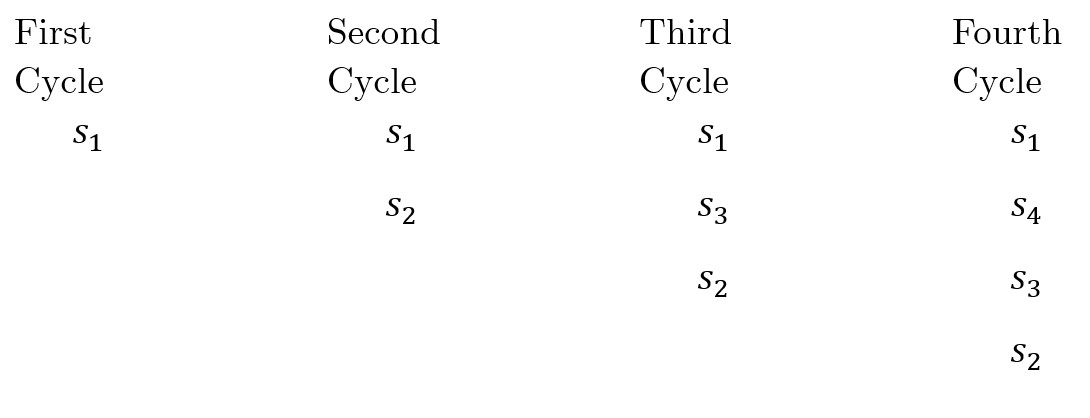} 
		
		\caption{The figure illustrates the pattern in updating the stabilizers with a set of measurements. Each column gives the  set of stabilizer generators at the beginning of the cycle.}
		\label{fig: init_floquet}
		
	\end{figure}

First, we construct a measurement sequence to initialize a Floquet code with three stabilizer generators. Consider the following sequence of measurements: 
\begin{equation*}
    M= \{ s_1,d_{s_1} d_{s_2}d_{s_3}, s_1 d_{s_2}d_{s_3}, d_{s_1}s_2 s_3, s_2\} 
\end{equation*}
Repeating this sequence of measurements gives the following sets of stabilizer generators at the end of each cycle, with $ S_i^{(0)}  $ indicating that it is a set of stabilizer generators after the $ i^{\mathrm{th}} $ cycle of measurements: $ S_1^{(0)} = \{s_1, s_2\}, S_2^{(0)} = \{s_1,s_2,s_3\}, S_3^{(0)}=\{s_1,s_2,s_3\},S_k^{(0)} = S_3^{(0)} $ for $ k>3 $. 
	To construct the sequence for 4 stabilizers, the idea is to insert a measurement sequence right before measuring $ s_2 $ in $ M $.  The set of stabilizer generators right before measuring $ s_2 $ is given by $ \{d_{s_1}s_2 s_3, d_{s_1} s_2, d_{s_1}s_3\} $ , which one can verify by the stabilizer update rule. The goal is to insert a sequence to perform the changes as illustrated in Figure  \ref{fig:seq_insert_floquet}.
	
		\begin{figure}[ht]
		\centering
		\includegraphics[width= 0.8 \linewidth]{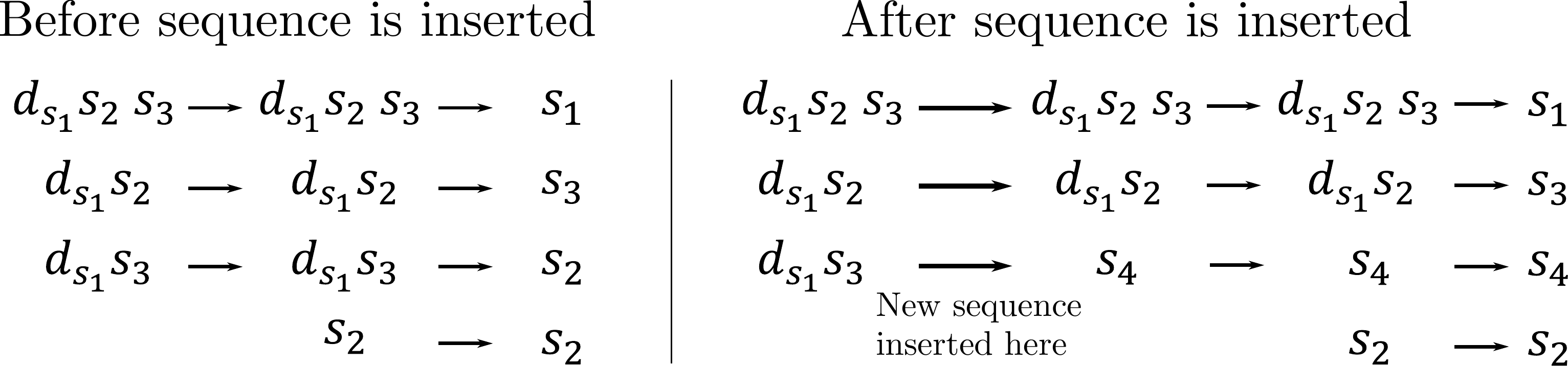} 
		
		\caption{The figure illustrates the effect of adding an additional sequence of measurements so that a new stabilizer generator is added by the end of the $ 3^{\mathrm{rd}} $ cycle. The new sequence is inserted into the original sequence $M$, at the position before $s_2$ is measured.}
		\label{fig:seq_insert_floquet}
		
	\end{figure}
	One can do this by making use of a sequence similar to $ M $: 
 \begin{equation*}
     M_3= \{ d_{s_1} d_{s_2}d_{s_3}, s_1 d_{s_2}d_{s_3}, d_{s_1}s_2 s_3,s_1\}
 \end{equation*}
This maps $ s_1 \to s_1, s_2 \to s_3 $  when $ n= 3 $. To do this for the $ 4^{\mathrm{th}} $ cycle, we replace the variables in $ M_3 $ with the elements from the stabilizer group right before where the new sequence will be inserted using this mapping: $ s_1 \to d_{s_1}s_2 s_3 $, $ s_2 \to d_{s_1}s_3 $, and $ s_3 \to s_4 $.  Using the updated variables, $ d_{s_1} $ will be mapped to $ d_{d_{s_1}s_2 s_3 } $ which anti-commutes with $ d_{s_1}s_2 s_3  $ but commutes with all the other generators in the group including $ d_{s_1}s_2  $. Thus, the new sequence consists of measurements that commute with $ d_{s_1}s_2 $, so it will not be changed by the update from this set of new measurements, when it is inserted in the location indicated in Figure \ref{fig:seq_insert_floquet}.
	
This trick can be applied recursively for more stabilizers. To extend to $ n $ stabilizers from $ n-1 $ stabilizers. Right before measuring the last measurement $ s_2 $, we replace $ s_1 $ with $ d_{s_1}s_2 s_3, s_2  $ with the stabilizer generator that is evolved from $ s_{n-1} $, and replace $ s_3 $ with $ s_n $. 
	
One can verify that inserting the new measurement sequence does not affect the previous cycles. Since $ \langle S_j^{(i)} \rangle \subseteq  \langle S_{j+1}^{(i)} \rangle$, the stabilizer generators before making the measurements in the new sequence are a subset of the current cycle, so the new sequence will act as identity to those generators. One can also verify that in the future cycles, $ s_1 $ will still be mapped to $ s_1 $ and $ s_i\to s_{i+1} $ because the new sequence maps $ s_{n-1}\to s_n $ while preserving the other mappings.
\end{proof}

\begin{example}
We provide an alternative construction below that looks deceivingly simple, using constant cycle size and also purely local measurements (in 1D), yet takes $O(n)$ time to initialize/learn the syndrome:

Round 1) Measure $X_1$

Round 2) Measure $X_2 X_3$, $X_6 X_7$, and so on ($X_{4i+2} X_{4i+3}$)

Round 3) Measure $Z_1 Z_2$, $Z_5 Z_6$, and so on ($Z_{4i+1} Z_{4i+2}$)

Round 4) Measure $X_4 X_5$, $X_8 X_9$, and so on ($X_{4i} X_{4i+1}$)

Round 5) Measure $Z_3 Z_4$, $Z_7 Z_8$, and so on ($Z_{4i+3} Z_{4i+4}$)

Round $k>5$) Repeat Round $((k-2) \mod 4) +2$.

We explicitly show the ISGs for $n =10$, in which case we have round 1, followed by cycles of 4 rounds of measurements given by: $\mathrm{Round\; 1}: \mathrm{Measure\;} X_2 X_3, X_6 X_7$. $\mathrm{Round\; 2}: \mathrm{Measure\;} Z_1Z_2, Z_5Z_6,Z_9Z_{10}$. $\mathrm{Round\; 3}: \mathrm{Measure\;} X_4X_5,X_8X_9$. $\mathrm{Round\; 4}: \mathrm{Measure\;} Z_3Z_4,Z_7Z_8$

\begin{eqnarray*}
    &\mathrm{Round\;1)}& X_1\\
    &\mathrm{Round\;2)}& X_1,X_2 X_3, X_6 X_7, \\
    &\mathrm{Round\;3)}& Z_1Z_2, Z_5Z_6,Z_9Z_{10}, X_1X_2X_3\\
    &\mathrm{Round\;4)}& X_1X_2X_3,Z_1Z_2,X_4X_5,X_8X_9\\
    &\mathrm{Round\;5)}&Z_1Z_2,Z_3Z_4,Z_7Z_8,X_1X_2X_3X_4X_5\\
    &\mathrm{Round\;6)}&X_1X_2X_3X_4X_5, X_2X_3,X_6X_7, Z_1Z_2Z_3Z_4\\
    &\mathrm{Round\;7)}& X_1X_2X_3X_4X_5X_6X_7, Z_1Z_2Z_3Z_4, Z_1Z_2, Z_5Z_6,Z_9Z_{10}\\
    &\mathrm{Round\;8)}& X_1X_2X_3X_4X_5X_6X_7,Z_1Z_2Z_3Z_4Z_5Z_6, X_4X_5,X_8X_9,Z_1Z_2 \\
    &\mathrm{Round\;9)}& X_1X_2X_3X_4X_5X_6X_7X_8X_9, Z_1Z_2Z_3Z_4Z_5Z_6, Z_3Z_4,Z_7Z_8, Z_1Z_2\\
    &\mathrm{Round\;10)}& X_1X_2X_3X_4X_5X_6X_7X_8X_9,Z_1Z_2Z_3Z_4Z_5Z_6Z_7Z_8, X_6X_7, Z_3Z_4,Z_1Z_2\\
    &\mathrm{Round\;11)}& X_1X_2X_3X_4X_5X_6X_7X_8X_9,Z_1Z_2Z_3Z_4Z_5Z_6Z_7Z_8, X_6X_7, X_2X_3, Z_1Z_2Z_3Z_4\\
    &\mathrm{Round\;12)}& X_1X_2X_3X_4X_5X_6X_7X_8X_9,Z_1Z_2Z_3Z_4Z_5Z_6Z_7Z_8, Z_1Z_2Z_3Z_4, Z_1Z_2,Z_5Z_6,Z_9Z_{10}\\
    &\mathrm{Round\;13)}&X_1X_2X_3X_4X_5X_6X_7X_8X_9,Z_1Z_2Z_3Z_4Z_5Z_6Z_7Z_8Z_9Z_{10}, Z_1Z_2Z_3Z_4Z_5Z_6, X_4X_5, X_8X_9, Z_1Z_2
\end{eqnarray*}

Notice that for the chain of $X_1X_2X_3X_4\cdots X_n$ to form, it takes $n$ rounds. Thus, the initialization requires at least $(n-1)/ 4 =O(n)$ cycles.
\end{example}
\begin{theorem}
It takes at most $k$ cycles to unmask all the stabilizers from an ISG, with $k$ given by the number of cycles required to initialize a Floquet code. 
\end{theorem}

\begin{proof}
We show that for an ISG $S_j^{(i)}$, it takes at most $k$ cycles to unmask, by applying the distance algorithm with $C= S_j^{(i)}$ and $V$ as the empty set. 

Updating $C$ and $V$ using step 2 of the distance algorithm, we see that it takes at most $k$ cycles for $V$ to be fully initialized. After initialization, $C \subset V$. Further, we know that $\langle C\rangle \cap \langle V\rangle $ gives a complete set of unmasked stabilizers. Thus, after at most $k$ cycles, all the stabilizer generators in $S_j^{(i)}$ are either permanently masked stabilizers or unmasked stabilizers.                                              
\end{proof}

\begin{corollary}
    For a Floquet code with $k=O(1)$ cycles to initialize and $m= O(1)$ measurements per cycle, then to fully determine the unmasked stabilizers for an ISG, it takes at most $m\cdot k = O(1)$ measurements. 
\end{corollary}

\section{Geometrically local dynamical code}
\label{sec: geometrically-local-codes}

The honeycomb code and other generalizations of Floquet codes are geometrically local codes. In this section, we consider some properties of dynamical codes (not necessarily Floquet codes with repeating cycles) for which the measurements are all geometrically local.

Based on our definition and discussion of distance from Section \ref{sec: dynamical code}, any ISG with geometrical locality will subject the dynamical code as a whole to its bound on distance and trade-offs in code parameters. Further, any neighboring pair of ISGs that form a geometrically local subsystem code will satisfy similar code parameter bounds. 

We summarize previous results on transversal gates, distance and trade-offs in codes with geometric locality: 

\textbf{Bounds on distance:}

\begin{theorem} [Bravyi and Terhal]\cite{bravyi_no-go_2009}
	
	If the generators of the gauge group $\mathcal{G}$ are geometrically local, then the distance satisfies $d \leq 3rL^{D-1}$, where $r$ is the width of the local gauges. In particular, local stabilizer codes have distance given by $O(L^{D-1})$.
\end{theorem}

\textbf{Bounds on trade-offs in code parameters}

\begin{theorem}[Bravyi, Poulin, Terhal]\cite{bravyi_tradeoffs_2010}
	
	Local stabilizer codes satisfy the following constrain:
	\begin{equation*}
		k \leq \frac{cn}{d^{\alpha}},\:\alpha = \frac{2}{D-1}
	\end{equation*}
	
\end{theorem}

\begin{theorem}[Bravyi]\cite{bravyi_subsystem_2011}

	For subsystem codes with gauge group that has local gauges, the following is satisfied in 2D:
	\begin{equation*}
		kd = O(n)
	\end{equation*}
	
\end{theorem}

\begin{example}
	The Hastings and Haah honeycomb code \cite{Hastings2021dynamically} has 2D ISGs that allow for fault tolerant Clifford gates, and have distance $O(\sqrt{n})$. It also satisfies the BPT bound for all ISGs. 
\end{example}

Outside of geometric locality, the connectivity of a code can constrain its code parameters:

\begin{theorem}[Baspin, Guruswami, Krishna, Li] \cite{baspin2023improved}

	For a code whose connectivity graph $G$ has separation profile $s_G(r)\leq O(r^c)$ for some $c \in (0,1]$, it holds that
	\begin{equation*}
		kd^{\frac{1-c^2}{c}}= O(n)
	\end{equation*}
	\begin{equation*}
		d=O(n^c)
	\end{equation*}

\end{theorem}

\textbf{Bounds on transversal gates:}
\begin{theorem}[Bravyi and K$\mathrm{\ddot{o}}$nig]\cite{bravyi_classification_2013}

    Suppose a unitary operator $U$ implementable by a constant-depth quantum circuit preserves the codespace $\mathcal{C}$ of a topological stabilizer code on a $D$-dimensional lattice, $D \geq 2$. Then the restriction of $U$ onto $\mathcal{C}$ implements an encoded gate from the set in the $D^{\mathrm{th}}$ level of Clifford hierarchy.  
\end{theorem}

\begin{theorem}[Pastawski and Yoshida] \label{PY bound on transversal gate}\cite{pastawski_fault-tolerant_2015}

    The above result applies to subsystem codes with generators of the gauge group $\mathcal{G}$ being geometrically local: Assuming that the code has a loss threshold $p_l > 0$ and a code distance $d= \Omega(\log^{1-1/D}(n))$, any locality-preserving logical unitary, fully supported on a $D$-dimensional region, has a logical action included in $\mathcal{C}^{(D)}$.
\end{theorem}

\subsection{Non-Clifford gates on 2D dynamical code?}
One important question in this paper is whether a dynamical code can offer new ways of implementing fault tolerant non-Clifford gate on a 2D lattice, which has previously been shown to be not possible for static geometrically local stabilizer or subsystem codes. Several open questions have been raised in this regard: One, dynamical code can allow for qubits to participate as ancilla in some rounds and physical qubits in others. It is unclear if this allows for more efficient implementations of non-Clifford gates. Two, we raise the question of whether it is possible to implement a logical gate by first performing unitaries to bring the code to a non-stabilizer state before making measurements to project it back to a (possibly different) codespace. 

In this section, we consider the property of long range connectivity: Unlike geometrically local stabilizer codes or subsystem codes, a dynamical code can allow for long range stabilizers that cannot be expressed as local gauges in certain rounds, as the long range stabilizers can be measured using gauges spread across multiple rounds. The Bacon Shor code, for instance, has non-local stabilizers measured in either choice of gauges, forming a subsystem code with local gauges. Unlike these codes, a dynamical code has multiple rounds of ISGs where neighboring ISGs may not form a geometrically local subsystem code. Long range stabilizers may not be measured across multiple rounds so a dynamical code can allow for some amount of non-local stabilizers, not present in the usual stabilizer framework. Further, one can also shuffle ancillas to connect some physical qubits on certain rounds of the code to obtain limited long-range connectivity or measure non-local stabilizers. 

We restrict the problem to the case where a non-Clifford gate can be directly implemented by a constant depth circuit on a particular round of ISG where some stabilizers are non-local, and this can take the code to either the same codespace or a different codespace, stabilized by another round of ISG. These two ISGs can be connected by a sequence of measurements. While we cannot rule out the existence of other ways for gate implementation on dynamical codes, it is natural to focus on gates that directly map between stabilizer codes. To keep error propagation under control, only constant depth circuits are considered for logical gates. As an example of 2D implementation with a non-Clifford gate, the doubled color code \cite{bravyi2015doubled} can implement a non-Clifford T gate by introducing additional ancillas for long range stabilizers and performing gauge fixing.

We show two separate results regarding the connection of long-range connectivity and logical gates from higher Clifford hierarchy. These can be seen as a generalization of the results by Pastawski and Yoshida (PY) to geometrically local stabilizer codes and subsystem codes with limited non-local stabilizers without local gauges supported on some physical qubits. First, we show that for both geometrically local stabilizer codes and geometrically local subsystem codes with some amount of long-range connectivity, if there exist logical $X$ and $Z$ representatives for a logical qubit that are fully supported on a region far enough from the set of qubits with long range connectivity, then one cannot find a transversal gate that implements a logical single qubit gate beyond $\mathcal{C}^{(D)}$ on that logical qubit. Secondly, we also show that if there are not too many qubits with long range connectivity, the code does not support any gates from higher Clifford hierarchy. 

These new results provide a lower bound on the number of qubits that must support long range stabilizers before one can ask the question of whether a code has a gate from a higher Clifford hierarchy. An intuition that can be gained from here is that in general having a small amount of connectivity will not improve the code's performance at supporting gates from higher Clifford hierarchy. One can show that if there is such an improvement, then it is limited to the logical operators that can be fully supported by the region given roughly by $Q$, the set of qubits with long range connectivity. This implies that the distance of the code for this logical qubit is upper bounded by roughly the size of $Q$, i.e. either a mostly geometrically local code has a good distance but no non-Clifford gate, or the code has a bad distance but supports a non-Clifford gate in 2D. This result also holds for all spatial dimensions. 

We build on PY's results and assume similar mild assumptions for the family of quantum codes, with the additional assumption that the codes also have some long-range stabilizers supported on $Q$ physical qubits. The proof borrows several ideas from  \cite{pastawski_fault-tolerant_2015}, although it is technically different and shows a different result.

The following theorems and definitions will be important for the subsequent proofs.

\begin{theorem}(Cleaning Lemma) \cite{bravyi_no-go_2009}\cite{pastawski_fault-tolerant_2015}
If a subset $R$ supports no logical operator (except the one with trivial action), then any logical operator $P$ can be cleaned from $R$.
\end{theorem}

\begin{theorem}(Union Lemma for stabilizer codes)\cite{bravyi_no-go_2009}\cite{pastawski_fault-tolerant_2015}
For a stabilizer code, let $R_1$ and $R_2$ be two disjoint sets of qubits. Suppose there exists a complete set of stabilizer generators $S$ such that the support of each generator overlaps with at most one of $\{R_1,R_2\}$. If $R_1$ and $R_2$ are bare cleanable, then the union $R_1 \cup R_2$ is also bare cleanable. 
	
\end{theorem}
\begin{theorem}(Union Lemma for subsystem codes)\cite{pastawski_fault-tolerant_2015}
For a subsystem code, let $R_1$ and $R_2$ be two disjoint sets of qubits. Suppose there exists a complete set of gauge group generators such that the support of each generator overlaps with at most one of $\{R_1,R_2\}$. If $R_1$ and $R_2$ are dressed cleanable, then the union $R_1 \cup R_2$ is also dressed cleanable.   
\end{theorem}

\begin{definition}\cite{pastawski_fault-tolerant_2015}
A region $R$ is bare-cleanable (dressed cleanable) if it supports no non-trivial dressed (bare) logical operators.
\end{definition}

\begin{definition}
The complement of a region $ M $ on a set of physical qubits is denoted by $ \overline{M} $.
\end{definition}

\begin{definition}
$\mathcal{B}(R, r)$ is an $r$-neighbourhood of a region $R$ which includes region $R$ and all particles within distance $r$ to it. The spread $s_U$ is defined as the smallest distance such that $\forall A: \mathrm{supp}(UAU^{\dagger})\subseteq \mathcal{B}(\mathrm{supp}(A), s_U)$.
\end{definition}

\begin{definition}
The boundary of a region $ M $, denoted by $ \partial M $, is defined as the smallest support of the set of stabilizer generators that have non-trivial overlap with both $ M $ and $ \overline{M} $ over all possible choices of basis set for $ \langle S\rangle $. Denote the boundary of $ M $ that lies outside of $ M $ as $ \partial_{+} M $ and the boundary of $ M $ that lies inside $ M $ as $ \partial_{-} M $. 
	
\end{definition}

\subsection{Geometrically local stabilizer codes with long range stabilizers} \label{subsec: geo_local_codes}
\begin{theorem}\label{main theorem: stabilizer codes}
Consider a family of stabilizer codes with geometrically local stabilizer generators embedded in $D$ spatial dimension, but support long range connectivity on some of the qubits. For any code $C$ in the family, let $ Q $ denote the set of qubits that supports long range connectivity and non-local stabilizers. Suppose the distance of the family of codes grows at least logarithmically with system size: $d=\Omega(\log(n))$. Then, given any constant depth circuit that is local with respect to the connectivity of the code $ C $ and implements a logical single qubit gate $ U $, there exists $ H = \mathcal{B}(Q, (2^{D+1}+1) s_U+c) $, where $ c $ is a constant that depends on $ s_U $, $ D $ and the radius of the local stabilizer generators of $ C $, such that if a pair of logical representatives $\{x_L,z_L\}$ for a logical qubit $ q_L $ is fully supported on $ \overline{H} $, then $ U $ must belong to $\mathcal{C}^{(D)}$ with respect to its logical action on $ q_L $.
	
\end{theorem}

\begin{proof}
Let $ C $ be a code in the family with $ n $ physical qubits, and let $ Q $ be the set of qubits that supports long range stabilizers. We want to construct $ D+1 $ bare cleanable regions, and use it to show that a constant depth circuit that respects the locality of the code $ C $ and implements the gate $ U $ must belong to $\mathcal{C}^{(D)}$ for a logical qubit $ q_L $, if it has a pair of logical representatives from the logical set $ \{x_L, y_L, z_L \}$ that can be fully supported on a region far enough from $ Q $. We denote $ r_{\mathrm{local}} $ as the maximum radius of local stabilizer generators in $C$.
	
\paragraph{Construction of D+1 regions} First we split the D dimensional lattice into unit cells with volume $ v_c = \alpha d$, where $ d $ is the distance and $ \alpha $ is a constant. For a fixed constant $ r = O(1) $, we can pick a ball of  radius $ r $ from each cell, such that each ball is spatially disjoint from the other balls and from the boundaries of the lattice by a constant physical distance that is at least $r_{\mathrm{local}}$ (See Figure \ref{fig: longrange}).
	
Let $ R_0' $ denote the union of the balls. $ R_0' $ may have overlap with $  H' := \mathcal{B}(Q,  (2^{D}+1) s_U+c')$, where $c'>2r_{\mathrm{local}}$. We deform $H'$ locally so that the balls that overlap with $ H' $ are now contained in $ H' $. Let $c_0$ denote the maximum physical distance of the deformation from the boundary of $H'$. Then, we set $c=c_0 +c'$ and define $H_1$ as $ H_1 := \mathcal{B}(Q,  (2^{D}+1) s_U+c)$, with $ H' \subseteq H_1 $. 
	
Let $ R_0 := R_0'\backslash H' $. This is the set of balls that are far away from $H'$. We want to show that $ R_0 $ forms a bare cleanable region. Since each ball is $ o(d) $, they are each bare cleanable. Further, we have picked the balls to be supported away from the boundaries of the unit cells, such that the balls are at least $2r_{\mathrm{local}}$ apart, so any local stabilizer generator overlaps with at most 1 ball and since $ R_0 $ does not overlap with $ H' $, any long range stabilizers that are only supported on $ Q $ are spatially disjoint from $ R_0 $. Thus, the Union Lemma can be applied to $ R_0 $, and we conclude that $ R_0 $ forms a bare cleanable region. 
	
Using $ R_0'$ as the set of mutually disjoint balls on the code, one can draw lines connecting these balls, then fatten the lines to form a covering for all the physical qubits, following the same construction in PY’s result. This gives a covering of the full lattice with $ R_m' $ for $ m \in [0, D] $. $ R_m' $ consists of $ m $ dimensional connected components. $ R_D' $ consists of $ D $ dimensional skewed cells. See Figure \ref{fig: longrange} for illustration. 

\begin{figure}[ht]
    \centering
    \includegraphics[width=\linewidth]{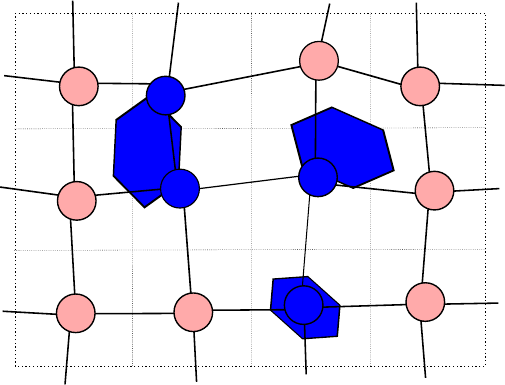} 
    
    \caption{An illustration of the construction of 3 regions $R_0,R_1,R_2$ for a 2D lattice. The unit cells are represented by the dotted lines in the background. The pink (light shade) and blue (dark shade) balls are picked from each cell such that they are spatially disjoint. The blue (dark shade) hexagons represent $H'$. $R_0'$ is the set of all the balls. $R_0$ is given by the set of pink balls that have no overlap with $H'$. After merging in the blue balls with $H'$, the entire blue region gives the new $H'$. } 
    \label{fig: longrange} 
\end{figure}

As the new regions can overlap with $H'$, we define the regions as follow: $ R_m := R_m'\backslash  H'$, for $1\leq m \leq D$. We show that each region $ R_m $ is bare cleanable. The volume of each connected component $ r_m' $ of  $ R_m' $ is at most $ O(\alpha d^{\frac{m}{D}})=O(d)$. Since any region with volume less than the distance $ d = \Omega(\log(n)) $ is bare cleanable, $ r_m'\backslash H'=O(d) $ is also a bare cleanable object. Since each $ r_m' $ is spatially disjoint, and each $ r_m'\backslash H'$ is also spatially disjoint. Supported away from the qubits with long range stabilizers, $ r_m'\backslash H'$ can also be cleaned using local stabilizers that have no overlap with other $ m$-dimensional objects. Thus, we can apply the Union Lemma to conclude that $ R_m $ is bare cleanable. 
	 
Given a constant depth circuit, we have to check that each region $R_m^+ := \mathcal{B}(R_m, 2^{m-1} s_U)$ is still bare cleanable for $ m>0 $, which is a condition required for the next part of this proof. For each connected component $ r_m' $ in $ R_m' $, any Pauli operator supported on $ r_m'\backslash H'$ is spread locally since it is at least $(2^{D}+1)s_U$ away from $Q$, the region that does support long range connectivity. Further, $\mathcal{B}(r_m'\backslash H', 2^{m-1} s_U)$ has no overlap with $Q$, so cleaning $ \mathcal{B}(r_m'\backslash H', 2^{m-1} s_U)$ does not spread Pauli operators supported on one of the $r_m'\backslash H'$ to other objects in $ R_m^+$ by any long range stabilizers. Since $ r_m'\backslash H' $ increases by a constant radius of $ 2^{m-1} s_U =O(1)$, $ \mathcal{B}(r_m'\backslash H', 2^{m-1} s_U) $ remains bare cleanable. By choosing $ r $ to be greater than $  2^{m} s_U$, each object of the form $ \mathcal{B}(r_m'\backslash H', 2^{m-1} s_U) $ is disconnected from the other objects in $ R_{m}^+$. Therefore, the union of $ R_m^+ $ is still bare cleanable.
	 
\paragraph{Proof that $ U $ lies in $ \mathcal{C}^{(D)} $ }
	 
Next we prove that $ U $ must be in $ \mathcal{C}^{(D)} $, with respect to the logical qubit $ q_L $. Suppose the logical representatives $\{x_L,z_L\}$ are fully supported on $\overline{H} := \overline{\mathcal{B}(H_1,2^D s_U)}$. The proof holds if we replace $H_1$ with $H'$ but for simplicity and better clarity we will use $H_1$ instead. 
	 
Without assuming that $ H$ is bare cleanable, we pick two arbitrary logical representatives $ L_1 $ and $ L_2 $ for the logical qubit $ q_L $ such that they are fully supported on $ \overline{\mathcal{B}(H_1,2^D s_U)} $. We show that $ U \in \mathcal{C}^{(2)} $ with respect to the logical qubit for $ D=2 $ before generalizing the proof to arbitrary dimension. For $D=2$, $ H_1 = \mathcal{B}(Q, 5 s_U+c), $ $  R_1^+ = \mathcal{B}(R_1,s_U) $, $R_2^+ = \mathcal{B}(R_2,2s_U) $
	 
$ L_1 $ can be cleaned from $ R_1^+ $, since $ R_1^+ $ is dressed cleanable. After cleaning, $L_1$ is supported on:
\begin{equation*}
    D_1 = \overline{\mathcal{B}(H_1,4 s_U)} \backslash R_1^+ \cup \partial_+ R_1^+
\end{equation*}
	 
$ L_2 $ can be cleaned from $ R_2^+ $, and is supported on:
\begin{equation*}
     D_2 = \overline{\mathcal{B}(H_1,4 s_U)} \backslash R_2^+ \cup \partial_+ R_2^+
\end{equation*}

Let $ L_1' := UL_1U^{\dagger}L_1^{\dagger} $ and $ L_2' = L_2L_1'L_2^{\dagger}L_1'^{\dagger} $. We show that $L_2'$ is an identity operator up to $\pm 1$ phase. This will imply that $ L_1' $ is a Pauli operator so $ U$ is at most a Clifford gate on $q_L$. 

$L_1'$ is supported fully on $\mathcal{B}(D_1, s_U)$ which is given by:
\begin{eqnarray*}
    &&\mathcal{B}(D_1, s_U) \\
    &=& \mathcal{B}(\overline{\mathcal{B}(H_1,4 s_U)} \backslash R_1^+ , s_U) \cup \mathcal{B}(\partial_+ R_1^+, s_U) \\
    &=&\overline{\mathcal{B}(H_1,3 s_U)} \backslash R_1 \cup \mathcal{B}(\partial_+ R_1^+, s_U) 
\end{eqnarray*}
The first term is in $R_0 \cup R_2$. The second term can be expanded as follows: 
\begin{eqnarray*}
    && \mathcal{B}(\partial_+ R_1^+, s_U) \\
    &=& \mathcal{B}(R_0 \cap \partial_+ R_1^+, s_U)\cup \mathcal{B}(R_2 \cap \partial_+ R_1^+, s_U) \cup  \mathcal{B}(H_1 \cap \partial_+ R_1^+, s_U)
\end{eqnarray*}

First two terms are in $R_0$ and $R_2$ respectively. The last term is $\mathcal{B}(\partial \mathcal{B}(Q,  4 s_U+c) \cap \partial_+ R_1^+, s_U) $, with a \textit{local} spread of constant $ s_U $, since the support is at least $ s_U $ away from $ Q $.

We find that $ L_1'$ is fully supported on $ R_0,R_2 $ and $\mathcal{B}(\partial \mathcal{B}(Q,  4 s_U+c) \cap \partial_+ R_1^+, s_U) $.
 	
$ L_2'$ is fully supported on the region that supports $L_1'$, by treating $L_2$ as a unitary operator acting on $L_1'$. Treating $L_1'$ as a unitary operator acting on $L_2$, we also obtain that $ L_2'$ is fully supported on $\mathcal{B}(D_2, 2 s_U)$, which is given by:
\begin{eqnarray*}
    &&\mathcal{B}(D_2, 2 s_U) \\
    &=& \mathcal{B}(\overline{\mathcal{B}(H_1,4 s_U)} \backslash R_2^+ , 2 s_U) \cup \mathcal{B}(\partial_+ R_2^+,2 s_U) \\
    &=&\overline{\mathcal{B}(H_1,2 s_U)} \backslash R_2 \cup \mathcal{B}(\partial_+ R_2^+, 2s_U) 
\end{eqnarray*}
The first term is contained within $R_0 \cup R_1$. The second term can be expanded as: 
\begin{eqnarray*}
    && \mathcal{B}(\partial_+ R_2^+, s_U) \\
    &=& \mathcal{B}(R_0 \cap \partial_+ R_2^+, 2s_U)\cup \mathcal{B}(R_1 \cap \partial_+ R_2^+, 2s_U) \cup  \mathcal{B}(H_1 \cap \partial_+ R_2^+, 2s_U)
\end{eqnarray*}
First two terms are in $R_0$ and $R_1$ respectively. The last term can be rewritten as $\partial \mathcal{B}(Q,  3 s_U+c) \cap \partial_+ R_2^+ $, with a spread of constant $ 2 s_U $ which is local since the support is at least $ s_U$ away from $ Q $.
 	
Thus, $ L_2' $ must be supported on $ R_0 $ and the intersection between the two different boundaries each with some constant spread. However this is upper bounded by the support on the union of $ R_0 $ and the overlap of $ R_1^+ $ and $ R_2^+ $ with a constant local spread of $ 3 s_U $. These boundary objects have volume at most $O(d^{\frac{D-1}{D}})$ by construction, since they have a constant width in one of the dimensions. Thus, each object is bare cleanable, and by construction, they are spatially disjoint objects. Further, they are supported only on the boundary of $H_1$, so they can be cleaned by local stabilizers. Thus, by the Union Lemma, they form a bare cleanable set. As the boundary objects have no overlap with $ R_0 $ since they must be fully supported on the boundary of $\mathcal{B}(Q,  4 s_U+c)$ and $\mathcal{B}(Q, 3 s_U+c)$ with a constant spread of $2s_U$ at most, which is disjoint from $ R_0 $ by construction, the union of $ R_0 $ and the boundaries is also bare cleanable. Hence, $ L_2' $ is an identity operator. This implies that $ L_1' $ is a Pauli operator so $ U$ is at most a Clifford gate on $q_L$. 

Next, we generalize this to $ D $ dimension with $D+1$ regions. Pick $ D $ arbitrary logical representatives that are fully supported on $ \overline{\mathcal{B}(H_1, 2^Ds_U)} $, where $H_1 = \mathcal{B}(Q,  (2^{D} +1)s_U+c)$ as defined earlier. 

Each $L_i$ is cleaned from $ R_i^+ $. The support of $L_i$ is given by:
\begin{equation*}
    D_i = \overline{\mathcal{B}(H_1,2^D s_U)} \backslash R_i^+ \cup \partial_+ R_i^+
\end{equation*}

Define $ L_1' = UL_1U^{\dagger}L_1^{\dagger} $ and $L_i' = L_iL_{i-1}'L_i^{\dagger}L_{i-1}'^{\dagger}  $ for $ i>1 $ recursively. 

Using similar arguments as in the $D =2$ case, one can show that $L_i'$ is supported on:
\begin{equation*}
    \bigcup_{j \neq i} R_j \cup \mathcal{B}(\partial_+ R_i^+ ,  2^{i-1}s_U)
\end{equation*}

$\mathcal{B}(\partial_+ R_i^+ ,  2^{i-1}s_U)$ can be written as: 

\begin{equation*}
    \mathcal{B}(\bigcup_{j}R_j \cap \partial_+ R_i^+ ,  2^{i-1}s_U) + \partial \mathcal{B}(Q,  (2^{D} +1-2^{i-1})s_U+c)\cap\partial_+ R_i^+ 
\end{equation*}

The second term in the above equation is a boundary term. One can show that $ L_{D}' $ must be supported on the union of $ R_0 $ and the intersection between the different boundary terms. The latter is upper bounded by $ \bigcap_{k \neq 0}R_k^+ $, with a spread of at most $  2^{D-1}s_U $, which consists of 1 dimensional objects with constant width in other dimensions. These objects are bare cleanable since they are of volume $O(d)$. It is also easy to verify that each bare cleanable object is spatially disjoint and is supported outside of $ Q $ and $ R_0 $. Thus, $ L_{D}' $ is bare cleanable and must be a trivial operator. We can then recursively deduce that $ L_{i}' \in \mathcal{C}^{(D-i)} $. Then, $ U \in \mathcal{C}^{(D)} $.

\end{proof}

\begin{corollary}

($H$ is bare cleanable) In the case where $ H = O(d) $ is bare cleanable, we can show that $ R_0 \cup H $ forms a bare cleanable region since each connected component is bare cleanable and the stabilizer generators only overlap with at most one connected component. For any logical qubit, it is bare cleanable from $ H $, so it can be fully supported on $ \overline{H} $, and it follows that any logical gate $ U $ must lie in $ \mathcal{C}^{(D)} $.
		
\end{corollary}

\begin{corollary}(Multiple logical qubits with representatives supported on $\overline{H}$) Suppose $Q_A$ is the set of qubits with logical representatives that can be fully supported on $\overline{H}$ and $Q_B$ is the set of qubits with logical representatives on $H$. Then for any unitary operator that can be implemented by a constant depth circuit and can be decomposed into a tensor product structure of the form $\mathcal{H}_{Q_A}\otimes \mathcal{H}_{Q_B}$, the logical gate acting on $\mathcal{H}_{Q_A}$ must be in $\mathcal{C}^{(D)}$.
    
\end{corollary}
\subsection{Subsystem codes with long range stabilizers}

For a subsystem code, we cannot employ the same proof from Subsection \ref{subsec: geo_local_codes} as the boundary terms between $R_i^+$ and $H$ are only dressed cleanable. However, we are able to prove a slightly weaker statement with the requirement that the pair of logical representatives are $O(d^{1/D})$ physical distance away from $Q$, unlike in the local stabilizer code case where this only needs to be a constant $O(1)$.

\begin{theorem}\label{main theorem: sub stabilizer codes}

Consider a family of subsystem codes with geometrically local gauge generators embedded in $D$ spatial dimension, but support long range connectivity on some of the qubits. For any code $C$ in the family, let $ Q $ denote the set of qubits that supports long range connectivity and non-local gauges. Suppose the code satisfies the following assumptions. 
\begin{enumerate}
    \item Finite loss threshold: $p_l >0$	
    \item Distance grows at least logarithmically with the system size: $d=\Omega(\log(n))$		
\end{enumerate}
Then, given any constant depth circuit that is local with respect to the connectivity of the code $ C $ and implements a logical single qubit gate $ U $, there exists $ H = \mathcal{B}(Q, c) $, where $ c $ is $ O( \alpha d^{\frac{1}{D}})$ and $\alpha$ a constant, such that if a pair of logical representatives $\{x_L,z_L\}$ for a logical qubit $ q_L $ is fully supported on $ \overline{H} $, then $ U $ must belong to $\mathcal{C}^{(D)}$ with respect to its logical action on $ q_L $.

\end{theorem}
	
\begin{proof}
		

Let $ C $ be the code of interest satisfying the assumptions in the theorem, and let $ d $ be the distance of the code. Let $ U $ be a unitary gate that is implemented by a constant depth circuit. Let $r_{\mathrm{local}}$ be the smallest radius of ball that is needed to cover a local gauge generator.

If $ C $ can be embedded in a $D$ dimensional lattice, we will construct $ D+1 $ regions, with $ R_0 $ as the bare cleanable region, and $ R_i, i \in [1,D]$, as dressed cleanable regions and use them to show that a unitary $ U $ implemented by a constant depth circuit that respects the connectivity on $ C $ must belong to $\mathcal{C}^{(D)}$, for the logical qubit that can be supported outside of $H = \mathcal{B}(Q, c) $, where $ c $ is given by $ O( \alpha d^{\frac{1}{D}}) $


\paragraph{Construction of D+1 regions} First we split the $ D$ dimensional lattice into unit cells with volume $ v_c = \alpha d $, where $ d $ is the distance and $ \alpha $ is a constant. 

By definition of the loss error threshold, each unit cell contains a ball of constant radius that is correctable with probability approaching unity as the system size $ n $ grows, using the following lemma.

\begin{lemma}
	The probability of finding a correctable ball of constant radius $ r $ in a unit cell approaches unity as $ n $, the number of qubits in a $ D $ dimensional lattice, increases, for a cell of volume of at least $ \log(n) $.
\end{lemma}

\begin{proof}
	Let the probability of a ball of radius $ r $ not being correctable be $ q $, which is a constant since each qubit has a fixed probability of being correctable. 
	The number of balls that can be packed in the subcell is roughly given by $k =\log(n) /(\beta r^D)  $ where $ \beta r^D $ is the volume of a ball.
	The probability of finding no ball being correctable is less than the probability of finding none of the  $\log(n) /(\beta r^D) $ balls being correctable, and this is given by 
	
	\begin{eqnarray*}
		&&\mathrm{prob} (k \mathrm{\: balls \: each \: not \: correctable\:})\\
		&=& q^{k}\\
		& =& e^{-\Omega(\log(n))} 
	\end{eqnarray*}
\end{proof}
This tends to $0$ at a rate that is polynomial in $n$. Hence, with probability approaching unity, each unit cell includes at least a ball of radius $ r $ in each cell, such that each ball is spatially disjoint from the other balls.

Let $ R_0' $ denote the union of these balls. Using $ R_0'$ as the set of mutually disjoint balls on the lattice, one can draw lines connecting these balls, then fatten the lines to form a covering for all the physical qubits, following the same construction in PY’s result. This gives a covering of the full lattice with $ R_m' $ for $ m \leq D $. 


However, if any connected component intersects with $ H_1= \mathcal{B}( Q,(2^{D}+1) s_U+c') $, where $c' > 2r_{\mathrm{local}}$, then we continuously deform $ H_1 $ so that these objects are contained within $ H_1$. This will fatten $ H_1 $ by at most $ c_0=(\alpha d)^{\frac{1}{D}} $. We consider the union of any connected component in any region $ R_i' $ that overlaps non-trivially with $ \partial_{+} H_1 $ but does not overlap with $ H_1$. Let these components be denoted by $H_2$. Then, we define $H = H_1\cup H_2$, as the union of these connected components that surround $H_1$ and $H_1$ itself. Next, define $ R_m := R_m' \backslash H_1$. $ R_0 $ is bare cleanable and each $ R_m $ remains dressed cleanable. The volume of $ H $ is upper bounded by $O(|Q|+ |Q|d^{\frac{1}{D}})$.


Given a constant depth circuit, we have to ensure that for each region $R_m^+ := \mathcal{B}(R_m, 2^{m-1} s_U)$ is still dressed cleanable, for $ m>0 $: For each connected component $ r_m' $ in $ R_m' $, $\mathcal{B}(r_m'\backslash H_1, 2^{m-1} s_U+c')$ has no overlap with $\mathcal{B}(Q,  s_U)$, so cleaning $ \mathcal{B}(r_m'\backslash H_1, 2^{m-1} s_U)$ by only local gauges does not spread Pauli operators to other connected components of $ R_m^+$ by any long range stabilizers since they are only supported on $ Q $, so it will remain geometrically local. Further, since $ r _m'\backslash H_1 $ increases by a constant size as $ 2^{m-1} s_U =O(1)$,  $ \mathcal{B}(r_m'\backslash H_1, 2^{m-1} s_U) $ remains dressed cleanable. By choosing $ r $ to be greater than $  2^{m} s_U$, $ \mathcal{B}(r_m'\backslash H_1, 2^{m-1} s_U) $ remains disconnected from the other connected components in $ R_m' $. Therefore, $ R_m^+ $ remains dressed cleanable.
	
\paragraph{Proof that $ U $ lies in $ \mathcal{C}^{(D)} $ }

Suppose a pair of dressed logical representatives $\{L_x,L_z\}$ for a logical qubit are supported on $ \overline{H} $, then one can show that if $ U $ is a unitary operator that implements a logical single qubit gate by a constant depth circuit, then $ U $ must be in $ \mathcal{C}^{(D)} $ with respect to its logical action on the logical qubit.
 
First, we show the result for $ m=3 $. Let $ L_1 $ and $ L_2 $ be two arbitrary logical operators that are fully supported on $ \overline{H} $. $  R_1^+ = \mathcal{B}(R_1,s_u) $, $R_2^+ = \mathcal{B}(R_1,2s_u) $.
 
$ L_1 $ is supported on the regions that are separated from $ H_1 $ by connected components given by $H\backslash H_1=H_2$. Thus, $ L_1 $ can be cleaned from $ R_1^+ $ so that it is $s_U$ away from $H_1$. This can be done by picking the width of the connected components $r_m'$ to be large enough but still $O(1)$. After cleaning, $L_1$ is supported on the region $ R_0 \backslash R_1^+ \cup R_2\backslash R_1^+ $. Similarly, $ L_2 $ can be cleaned from $ R_2^+ $, and is fully supported on $ R_0\backslash R_2^+ \cup R_1\backslash R_2^+$. 
 
 $ L_1' = UL_1U^{\dagger}L_1^{\dagger} $ is supported on $ R_0\cup R_2 $. $ L_2' = L_2L_1'L_2^{\dagger}L_1'^{\dagger} $ is supported on $ R_0\cup R_1 $. From the commutation relation, $ L_2' $ is also supported on $ R_0\cup R_2 $. Thus, $ L_2' $ must be supported on only $ R_0 $ which is a bare cleanable region. $L_2'$ must be an Identity operator, up to some phases. This implies that $ L_1'  $ is a Pauli operator and $ U$ is at most a logical Clifford operator.

 We generalize to $ D $ dimension: First, pick $ D$ logical representatives that are fully supported on $ \overline{H} $, and label them $ L_i, i\in [1,D]$. $\forall i, L_i $  is cleaned from $ R_i^+ $, and is supported on the union of  $ R_k\backslash R_i^+, k \neq i, k \in [0, D]$. Here, we again make use of the fact that $H_1$ is fully surrounded by connected components, so we can make the components wide enough such that after cleaning, $L_i$ is supported on other dressed cleanable regions at least $2^{i-1}s_U$ away from $H_1$. Define $ L_1' = UL_1U^{\dagger}L_1^{\dagger} $ and define $L_i'$ recursively: $L_i' = L_iL_{i-1}'L_i^{\dagger}L_{i-1}'^{\dagger}  $ for $ 2 \leq i\leq D $. 
 
 Lastly, one can show that $ L_{D}' $ must be supported only on $ R_0 $. We prove this by induction on $i$. 

 Suppose $L_i'$ is supported on $R_0 \cup \bigcup_{m\geq i+1}R_m$. Since $L_{i+1}$ is supported at least $2^{i}s_U$ away from $H_1$, $\mathcal{B}(\mathrm{supp}(L_{i+1}), 2^{i}s_U)$ is supported on $\overline{H_1}$, $L_{i+1}'$ is supported on $R_0 \cup \bigcup_{m\neq i+1}R_m$. Since $L_{i+1}' = L_i'L_{i+1}L_i'^{\dagger}L_{i+1}^{\dagger} $, $L_{i+1}'$ is supported only on $\mathrm{supp}(L_i')$, so it is not supported on $\bigcup_{k< i+1}R_k$. Thus, $L_{i+1}'$ is supported on $R_0 \cup \bigcup_{m>i+1}R_m$.
 
 Thus, $ L_{D}' $ is a trivial logical operator since it is only supported on $R_0$. We can then recursively deduce that $ L_{i}' \in \mathcal{C}^{(D-i)} $. Then $ U \in \mathcal{C}^{(D)} $, and this concludes the proof. 


\end{proof}

 \begin{corollary}
 ($H$ is $O(d)$) In the case if we can show that $ H = O(d)$ for any code $ C $, then $ H $ cannot support a dressed logical operator, so it is a bare cleanable region. Then, we can always find logical representatives for any logical qubit so that they are fully supported on $ \overline{H} $. Therefore, any logical gate $U$ for a code with $ H =O(d)$ is an element in $ \mathcal{C}^{(D)}$.
 
 \end{corollary}

For a 2D geometrically local stabilizer code with some qubits to support long range stabilizers to implement a non-Clifford gate, our results show a lower bound of $ |Q|=O(d) $ number of qubits that must have long range connectivity. Further, if any pair of logical representatives can be supported a constant physical distance away from $ Q $, then the logical qubit for these pair of representatives does not have a transversal non-Clifford logical single qubit gate. 

This means that in order for a logical qubit to have a transversal non-Clifford single qubit gate, all representatives of a logical operator must have a non-trivial overlap with $\mathcal{B}( Q,c) $, where $ c $ is some constant of $ O(1) $. By the cleaning lemma, a logical operator representative of the logical qubit can be fully supported on $ \mathcal{B}( Q,c) $. One can then conclude that the distance for the logical qubit is small as it is restricted to $O(|\mathcal{B}( Q,c)|)$ qubits.

Interestingly, our result is related to a result in \cite{baspin_quantifying_2022}. One observation in \cite{baspin_quantifying_2022} is that to create a code with distance $d+\epsilon$ that breaks the BPT bound, it will require $\Omega(d+\epsilon)$ edges. Thus, the following holds: 

\begin{theorem}
    If the support of long range stabilizers is only $O(d)$, the distance of the code does not improve. Thus, a geometrically local code with $O(d)$ qubits supporting long range stabilizers must satisfy the BPT bound. 
\end{theorem}

We showed that if the support of long range stabilizers is only $ O(d) $, the logical gates that can be implemented using constant depth circuit must be contained in $ \mathcal{C}^{(D)} $, where $ D $ is the dimension of the lattice. 

The 2D doubled color code is an example of a geometrically local subsystem code that has non-Clifford gate in $ 2D $.

\begin{example}[2D doubled color code] \cite{bravyi2015doubled}
    The 2D doubled color code allows for a transversal implementation of T gate. The embedding gives $O(t^2)$ qubits that support long range gauge generators for a code with distance $d = 2t+1$ and $O(t^3)$ physical qubits. This gives an upper bound on the number of qubits that require long range support in a 2D geometrically local subsystem code setting in order to implement a non-Clifford gate. 
\end{example}
\section{Discussion}

A dynamical code can be constructed by a sequence of measurements. We can obtain the unmasked distance for each ISG and take the minimum to obtain an upper bound on the error correcting capacity of the dynamical code. This bound is applicable to a general circuit encompassing both unitaries and measurements.

However, there are certain challenges not addressed in this paper. First, our analysis excludes measurement errors, an important consideration when building the measurement sequence if we still want to protect the encoded information in the presence of measurement errors. We anticipate that future research will address this in the context of fault tolerance for dynamical codes. 

Additionally, exploring the initialization of a dynamical code and incorporating that into the construction of new examples could be of interest. Another concern for a dynamical code is that the last few rounds may not have complete measurements, so the distance of the ISGs in the last few rounds may not be as desired.  

Lastly, we have identified interesting theoretical lower and upper bounds on the number of qubits that must support long range stabilizers in a $ 2D $ embedding of a quantum code in order for it to support a non-Clifford gate fault tolerantly. Exploring this limit may lead to an interesting direction, potentially yielding a code that strikes an optimal balance between code distance, geometric locality and support for fault tolerant non-Clifford gates. 

\section*{Acknowledgments}
The authors would like to thank Zi-Wen Liu, Beni Yoshida, Noah Berthusen, Nathanan Tantivasadakarn, Christophe Vuillot, Margarita Davydova, Narayanan Rengaswamy, Ali Fahimniya, Hossein Dehghani and Prakhar Gupta for their helpful discussions. D.G is partially supported by the National Science Foundation (RQS QLCI grant OMA-2120575).

\medskip

\begin{thebibliography}{10}

\bibitem{Hastings2021dynamically}
Matthew~B. Hastings and Jeongwan Haah.
\newblock ``Dynamically {G}enerated {L}ogical {Q}ubits''.
\newblock \href{https://dx.doi.org/10.22331/q-2021-10-19-564}{{Quantum} {\bf
  5}, 564}~(2021).

\bibitem{haah_boundaries_2022}
Jeongwan Haah and Matthew~B. Hastings.
\newblock ``Boundaries for the {Honeycomb} {Code}''.
\newblock \href{https://dx.doi.org/10.22331/q-2022-04-21-693}{Quantum {\bf 6},
  693}~(2022).


\bibitem{vuillot2021planar}
Christophe Vuillot.
\newblock ``Planar floquet codes''~(2021).
\newblock \href{https://doi.org/10.48550/arXiv.2110.05348}{arXiv:2110.05348}.

\bibitem{kesselring2022anyon}
Markus~S. Kesselring, Julio C.~Magdalena de~la Fuente, Felix Thomsen, Jens
  Eisert, Stephen~D. Bartlett, and Benjamin~J. Brown.
\newblock ``Anyon condensation and the color code''~(2024).
\newblock \href{https://doi.org/10.1103/PRXQuantum.5.010342}{PRX Quantum {\bf 5}, 010342}.

\bibitem{ellison2023floquet}
Tyler~D. Ellison, Joseph Sullivan, and Arpit Dua.
\newblock ``Floquet codes with a twist''~(2023).
\newblock \href{https://doi.org/10.48550/arXiv.2306.08027}{arXiv:2306.08027}.


\bibitem{fahimniya2023faulttolerant}
Ali Fahimniya, Hossein Dehghani, Kishor Bharti, Sheryl Mathew, Alicia~J.
  Kollár, Alexey~V. Gorshkov, and Michael~J. Gullans.
\newblock ``Fault-tolerant hyperbolic floquet quantum error correcting codes''~(2025).
\newblock \href{https://doi.org/10.22331/q-2025-09-05-1849}{Quantum {\bf 9}, 1849}.

\bibitem{higgott2023constructions}
Oscar Higgott and Nikolas~P. Breuckmann.
\newblock ``Constructions and performance of hyperbolic and semi-hyperbolic
  floquet codes''~(2024).
\newblock \href{https://doi.org/10.1103/PRXQuantum.5.040327}{PRX Quantum {\bf 5}, 040327}.


\bibitem{PhysRevB.108.205116}
Zhehao Zhang, David Aasen, and Sagar Vijay.
\newblock ``$x$-cube floquet code: A dynamical quantum error correcting code
  with a subextensive number of logical qubits''.
\newblock \href{https://dx.doi.org/10.1103/PhysRevB.108.205116}{Phys. Rev. B
  {\bf 108}, 205116}~(2023).

\bibitem{bauer2023topological}
Andreas Bauer.
\newblock ``Topological error correcting processes from fixed-point path
  integrals''~(2024).
\newblock \href{https://doi.org/10.22331/q-2024-03-20-1288}{Quantum {\bf 8}, 1288}.


\bibitem{davydova_floquet_2023}
Margarita Davydova, Nathanan Tantivasadakarn, and Shankar Balasubramanian.
\newblock ``Floquet codes without parent subsystem codes''.
\newblock \href{https://dx.doi.org/10.1103/PRXQuantum.4.020341}{PRX Quantum
  {\bf 4}, 020341}~(2023).

\bibitem{Gidney_2021}
Craig Gidney, Michael Newman, Austin Fowler, and Michael Broughton.
\newblock ``A fault-tolerant honeycomb memory''.
\newblock \href{https://dx.doi.org/10.22331/q-2021-12-20-605}{Quantum {\bf 5},
  605}~(2021).

\bibitem{PRXQuantum.4.010310}
Adam Paetznick, Christina Knapp, Nicolas Delfosse, Bela Bauer, Jeongwan Haah,
  Matthew~B. Hastings, and Marcus~P. da~Silva.
\newblock ``Performance of planar floquet codes with majorana-based qubits''.
\newblock \href{https://dx.doi.org/10.1103/PRXQuantum.4.010310}{PRX Quantum
  {\bf 4}, 010310}~(2023).

\bibitem{aasen2023measurement}
David Aasen, Jeongwan Haah, Zhi Li, and Roger S.~K. Mong.
\newblock ``Measurement quantum cellular automata and anomalies in floquet
  codes''~(2023).
\newblock \href{https://doi.org/10.48550/arXiv.2304.01277}{arXiv:2304.01277}.


\bibitem{aasen_adiabatic_2022}
David Aasen, Zhenghan Wang, and Matthew~B. Hastings.
\newblock ``Adiabatic paths of {Hamiltonians}, symmetries of topological order,
  and automorphism codes''.
\newblock \href{https://dx.doi.org/10.1103/PhysRevB.106.085122}{Physical Review
  B {\bf 106}, 085122}~(2022).
  
\bibitem{sullivan2023floquet}
Joseph Sullivan, Rui Wen, and Andrew~C. Potter.
\newblock ``Floquet codes and phases in twist-defect networks''~(2023).
\newblock \href{https://doi.org/10.1103/PhysRevB.108.195134}{Phys. Rev. B {\bf 108}, 195134}.

\bibitem{alam2024}
M.~Sohaib Alam and Eleanor Rieffel.
\newblock ``Dynamical logical qubits in the bacon-shor code''~(2025).
\newblock \href{https://doi.org/10.1103/nfxv-3dp7}{Phys. Rev. A {\bf 112}, 022436}.


\bibitem{xu2025}
Yichen Xu and Arpit Dua.
\newblock ``Fault-tolerant protocols through spacetime concatenation''~(2025).
\newblock \href{https://doi.org/10.48550/arXiv.2504.08918}{arXiv:2504.08918}.


\bibitem{huang2018transversal}
Cupjin Huang and Michael Newman.
\newblock ``Transversal switching between generic stabilizer codes''~(2018).
\newblock \href{https://doi.org/10.48550/arXiv.1709.09282}{arXiv:1709.09282}.

\bibitem{kubica_universal_2015}
Aleksander Kubica and Michael~E. Beverland.
\newblock ``Universal transversal gates with color codes - a simplified
  approach''.
\newblock \href{https://dx.doi.org/10.1103/PhysRevA.91.032330}{Physical Review
  A {\bf 91}, 032330}~(2015).

\bibitem{anderson_fault-tolerant_2014}
Jonas~T. Anderson, Guillaume Duclos-Cianci, and David Poulin.
\newblock ``Fault-{Tolerant} {Conversion} between the {Steane} and
  {Reed}-{Muller} {Quantum} {Codes}''.
\newblock \href{https://dx.doi.org/10.1103/PhysRevLett.113.080501}{Physical
  Review Letters {\bf 113}, 080501}~(2014).

\bibitem{eastin_restrictions_2009}
Bryan Eastin and Emanuel Knill.
\newblock ``Restrictions on {Transversal} {Encoded} {Quantum} {Gate} {Sets}''.
\newblock \href{https://dx.doi.org/10.1103/PhysRevLett.102.110502}{Physical
  Review Letters {\bf 102}, 110502}~(2009).

\bibitem{bravyi_subsystem_2011}
Sergey Bravyi.
\newblock ``Subsystem codes with spatially local generators''.
\newblock \href{https://dx.doi.org/10.1103/PhysRevA.83.012320}{Physical Review
  A {\bf 83}, 012320}~(2011).

\bibitem{gottesman2022opportunities}
Daniel Gottesman.
\newblock ``Opportunities and challenges in fault-tolerant quantum
  computation''~(2022).
\newblock \href{https://doi.org/10.48550/arXiv.2210.15844}{arXiv:2210.15844}.

\bibitem{delfosse2023spacetime}
Nicolas Delfosse and Adam Paetznick.
\newblock ``Spacetime codes of clifford circuits''~(2023).
\newblock \href{https://doi.org/10.48550/arXiv.2304.05943}{arXiv:2304.05943}.


\bibitem{bacon_sparse_2017}
Dave Bacon, Steven~T. Flammia, Aram~W. Harrow, and Jonathan Shi.
\newblock ``Sparse {Quantum} {Codes} from {Quantum} {Circuits}''.
\newblock \href{https://dx.doi.org/10.1109/TIT.2017.2663199}{IEEE Transactions
  on Information Theory {\bf 63}, 2464--2479}~(2017).

\bibitem{derks2024}
Peter-Jan H.~S. Derks, Alex Townsend-Teague, Ansgar~G. Burchards, and Jens
  Eisert.
\newblock ``Designing fault-tolerant circuits using detector error
  models''~(2024).
\newblock \href{https://doi.org/10.48550/arXiv.2407.13826}{arXiv:2407.13826}.


\bibitem{beverland2024faulttolerance}
Michael~E. Beverland, Shilin Huang, and Vadym Kliuchnikov.
\newblock ``Fault tolerance of stabilizer channels''~(2024).
\newblock \href{https://doi.org/10.48550/arXiv.2401.12017}{arXiv:2401.12017}.

\bibitem{bravyi_no-go_2009}
Sergey Bravyi and Barbara Terhal.
\newblock ``A no-go theorem for a two-dimensional self-correcting quantum
  memory based on stabilizer codes''.
\newblock \href{https://dx.doi.org/10.1088/1367-2630/11/4/043029}{New Journal
  of Physics {\bf 11}, 043029}~(2009).

\bibitem{bravyi_classification_2013}
Sergey Bravyi and Robert Koenig.
\newblock ``Classification of topologically protected gates for local
  stabilizer codes''.
\newblock \href{https://dx.doi.org/10.1103/PhysRevLett.110.170503}{Physical
  Review Letters {\bf 110}, 170503}~(2013).

\bibitem{pastawski_fault-tolerant_2015}
Fernando Pastawski and Beni Yoshida.
\newblock ``Fault-tolerant logical gates in quantum error-correcting codes''.
\newblock \href{https://dx.doi.org/10.1103/PhysRevA.91.012305}{Physical Review
  A {\bf 91}, 012305}~(2015).

\bibitem{Bacon_2006}
Dave Bacon.
\newblock ``Operator quantum error-correcting subsystems for self-correcting
  quantum memories''.
\newblock \href{https://dx.doi.org/10.1103/physreva.73.012340}{Physical Review
  A{\bf 73}}~(2006).

\bibitem{bravyi2013subsystem}
Sergey Bravyi, Guillaume Duclos-Cianci, David Poulin, and Martin Suchara.
\newblock ``Subsystem surface codes with three-qubit check operators''~(2013).
\newblock \href{https://doi.org/10.48550/arXiv.1207.1443}{arXiv:1207.1443}.


\bibitem{davydova2023quantum}
Margarita Davydova, Nathanan Tantivasadakarn, Shankar Balasubramanian, and
  David Aasen.
\newblock ``Quantum computation from dynamic automorphism codes''~(2024).
\newblock \href{https://doi.org/10.22331/q-2024-08-27-1448}{Quantum {\bf 8}, 1448}.


\bibitem{bravyi_tradeoffs_2010}
Sergey Bravyi, David Poulin, and Barbara Terhal.
\newblock ``Tradeoffs for {Reliable} {Quantum} {Information} {Storage} in {2D}
  {Systems}''.
\newblock \href{https://dx.doi.org/10.1103/PhysRevLett.104.050503}{Physical
  Review Letters {\bf 104}, 050503}~(2010).


\bibitem{baspin2023improved}
Nouédyn Baspin, Venkatesan Guruswami, Anirudh Krishna, and Ray Li.
\newblock ``Improved rate-distance trade-offs for quantum codes with restricted
  connectivity''~(2023).
\newblock \href{https://doi.org/10.1088/2058-9565/ad8370}{Quantum Science and Technology {\bf 9}, 045024}~(2024).


\bibitem{bravyi2015doubled}
Sergey Bravyi and Andrew Cross.
\newblock ``Doubled color codes''~(2015).
\newblock \href{https://doi.org/10.48550/arXiv.1509.03239}{arXiv:1509.03239}.

\bibitem{baspin_quantifying_2022}
Nouédyn Baspin and Anirudh Krishna.
\newblock ``Quantifying nonlocality: how outperforming local quantum codes is
  expensive''.
\newblock \href{https://dx.doi.org/10.1103/PhysRevLett.129.050505}{Physical
  Review Letters {\bf 129}, 050505}~(2022).

\end{thebibliography}

\end{document}